\documentclass[%
 reprint,
 superscriptaddress,
 nofootinbib,
 amsmath,amssymb,
 aps,
]{revtex4-2}

\usepackage[margin=1in]{geometry}
\usepackage[utf8]{inputenc}
\usepackage[T1]{fontenc}
\usepackage{amsmath}
\usepackage{amssymb}
\usepackage{amsthm}
\usepackage{array}
\usepackage{physics}
\usepackage{titlesec}
\usepackage{stmaryrd}
\usepackage{soul}
\usepackage[colorlinks=true,citecolor=blue,linkcolor=magenta]{hyperref}

\newtheorem{lemma}{Lemma}

\usepackage{thm-restate}
\usepackage{thmtools}
\usepackage{graphicx}

\newtheorem{remark}{Remark}

\usepackage[nameinlink]{cleveref} %
\crefname{section}{Section}{Sections}
\crefname{equation}{Equation}{Equations}
\crefname{figure}{Figure}{Figures}
\crefname{table}{Table}{Tables}
\crefname{appendix}{Supplementary section}{Supplementary section}
\crefname{theorem}{Theorem}{Theorems}
\crefname{thm}{Theorem}{Theorems}
\crefname{cor}{Corollary}{Corollaries}
\crefname{lemma}{Lemma}{Lemmas}
\crefname{proposition}{Proposition}{Propositions}
\crefname{definition}{Definition}{Definitions}
\crefname{algorithm}{Algorithm}{Algorithms}
\crefname{remark}{Remark}{Remarks}
\let\autoref\cref

\newcommand{\RR}{\mathbb{R}}
\newcommand{\EE}{\mathbb{E}}
\newcommand{\VV}{\mathbb{V}}

\newcommand{\PP}{\mathbb{P}}
\newcommand{\Sum}{\displaystyle\sum}
\newcommand{\Frac}{\displaystyle\frac}

\newcommand{\by}{\boldsymbol{y}}

\newcommand{\bPhi}{\boldsymbol{\Phi}}
\newcommand{\bmnls}{\beta_{\mathrm{MNLS}}}
\newcommand{\fmnls}{f_{\mathrm{MNLS}}}

\usepackage{xcolor}

\begin{document}

\preprint{APS/123-QED}

\title{When Quantum and Classical Models Disagree: Learning Beyond Minimum Norm Least Square}

\author{Slimane Thabet }
\email{slimane.thabet@gmail.com}

\affiliation{LIP6, CNRS, Sorbonne Université, Paris, France}
\affiliation{PASQAL SAS, Massy, France}
\author{Léo Monbroussou}
\affiliation{LIP6, CNRS, Sorbonne Université}
\affiliation{Naval Group, 83190 Ollioules, France}
\author{Eliott Z. Mamon}
\affiliation{LIP6, CNRS, Sorbonne Université, Paris, France}
\author{Jonas Landman}
\affiliation{School of Informatics, University of Edinburgh, Edinburgh, United Kingdom}
\affiliation{DIENS, École Normale Supérieure, PSL University, CNRS, INRIA, Paris, France}

\date{\today}
\begin{abstract}
    Quantum Machine Learning algorithms based on Variational Quantum Circuits (VQCs) are important candidates for useful application of quantum computing. It is known that a VQC is a linear model in a feature space determined by its architecture. Such models can be compared to classical ones using various sets of tools, and surrogate models designed to classically approximate their results were proposed. At the same time, quantum advantages for learning tasks have been proven in the case of discrete data distributions and cryptography primitives. In this work, we propose a  framework to avoid Random Feature approximation techniques. Using previous results, we establish conditions on the weight vectors of the quantum models that are necessary to avoid these dequantization methods. We show that this theory is compatible with previously proven quantum advantages on discrete inputs, and provides examples of advantages for continuous inputs. This separation is connected to large weight vector norm, and we suggest that this can only happen with a high dimensional feature map. Our results demonstrate that it is possible to design quantum models that cannot be classically approximated with good generalization. In addition, we provide a method to verify that the necessary condition is respected for a quantum model. Finally, we discuss how concentration issues must be considered to design such instances. We expect that our work will be a starting point to design near-term quantum models that avoid dequantization methods by ensuring non-classical convergence properties, and to identify existing quantum models that can be classically approximated.
\end{abstract}

\maketitle

\section{Introduction}

Machine learning is a heavily explored area in the search of applications for quantum computers \cite{schuld2018supervised, cerezo2022challenges, biamonte2017quantum}. Some work has focused on using quantum computers as hardware accelerators of classical machine learning routines \cite{monbroussou2024QCNN, kerenidis2019q, kerenidis2016quantum}, mainly leveraging quantum linear algebra protocols \cite{harrow2009quantum, gilyen2019quantum}. Another heavily explored area is the use of variational quantum circuits (VQCs) \cite{cerezo2021variational} to learn some functions of the data \cite{schuld2020circuit}. The initial ideas of variational quantum machine learning (QML) research \cite{schuld2020circuit, havlivcek2019supervised} were that the advantage of quantum computing for machine learning would be to look for models in high dimensional feature spaces, exponentially larger than the initial dimension of the data, and the size of the dataset. %

\begin{figure*}
\begin{center}
\includegraphics[width=0.95\linewidth]{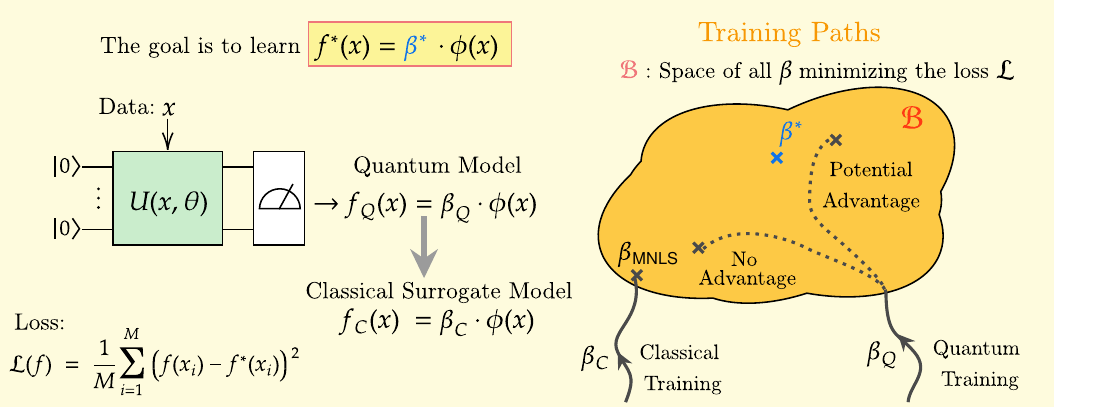}
\caption{A quantum model $f_Q(x) = \beta_Q(\theta)\cdot\phi(x)$ is trained by optimizing its weight vector $\beta_Q(\theta)$. If one can train a surrogate model on a classical computer, using the same (or approximated) feature map $\phi(x)$, it would constitute an obstacle to quantum advantage. It has been shown that during classical linear regression, the weight vector converges towards a specific point $\beta_{\mathrm{MNLS}}$ called the \textit{minimum norm least squares} estimator. Ensuring that the quantum weight vector $\beta_Q$ converges far from $\beta_{\mathrm{MNLS}}$ is therefore a necessary condition to avoid such dequantization.}
\label{fig:Introduction}
\end{center}
\vspace{-2em}
\end{figure*}

It is known that VQCs are linear models in some \textit{feature space} where each coordinate can be computed classically \cite{schuld2021effect}.
At first sight, if this feature map can be explicitly computed \cite{schuld2021supervised} classically, one may wonder what is the interest of searching the best parameters of the quantum circuit instead of performing classically a linear regression on the same feature map, using a so called \textit{classical surrogate} model \cite{schreiber2023classical}. Even when the feature space is growing exponentially with the number of qubits, and is thus too large to be computed classically, methods exist to reduce its dimension by random sampling, or Random Feature (RF) regression, realizing approximated classical models \cite{landman2022classically, sweke2023potential, sahebi2025dequantization}. 

Knowing when these classical models can mimic the quantum ones is crucial to understand potential quantum advantage in such learning tasks. 
As related work, in \cite{you2023analyzing}, the authors analyse the optimization dynamics of quantum neural networks and conclude that they are different from the neural tangent kernel. They study in detail the convergence rate of the respective methods, but do not study the actual solutions reached. Using feature maps with parity functions on bitstrings, the authors in \cite{jerbi2023quantum} study the fact that variational circuits can converge to a different solution than the kernel ridge regression. They point out that there exists functions that are learnable with VQCs but that require exponentially more resources to learn with the kernel associated with the feature map. 
In this work, we present necessary conditions for a quantum model to avoid  Random Feature (RF) based dequantization, and we highlight that this could be only satisfied for high dimensional feature maps. We present an extension of the RF approximation techniques that could be applied to any variational circuits with Hamiltonian encoding, with our without classical pre-processing of the classical input. Our study can be applied to any quantum circuits with continuous or discrete inputs, and propose conditions that guarantee a quantum model to remain \textit{far} from its equivalent classical model. For that, we use the important fact that classical linear regression tends to make the optimized weight vector converge towards a specific solution called \textit{minimum norm least square} (MNLS) estimator. Our study focuses on showing when the quantum weight vector doesn't possess the same bias. \autoref{fig:Introduction} summarizes our methodology. 

Those conditions are then analyzed for several usual frameworks and architectures, showing that our methodology can be seen as a new tool for one to rule out some quantum circuits of their choice. 
Furthermore, we propose a method that allow in some cases for the experimental verification of these conditions.
Using Weingarten calculus, we show that some proposed quantum models can be far from the MNLS. We also show that cryptographic examples with discrete input that were shown far from their classical counterparts \cite{jerbi2023quantum, jerbi2023shadows, liu2021rigorous, gyurik2023exponential} respect our condition on the weight vector norm. In addition, we study the link between these dequantization schemes and concentration, another crucial issue of quantum circuits. We prove that there should exist a family of models with continuous inputs that avoids both of these problems.

\section{Results}

\subsection{Setup and notations}
We consider general forms of quantum machine learning models, or \textit{quantum models} that can be expressed as
\begin{equation}
    f_Q(x) = \text{Tr}(U(x;\: \theta) ^\dagger O U(x;\: \theta) |0^n\rangle\langle 0^n|)
\end{equation}
where $U(x;\: \theta)$ is a unitary dependent on the input data $x \in \RR^d$ and trainable parameters $\theta$.
We consider that $U$ can be expressed as 
\begin{equation}
    U(x;\: \theta) = \prod_{\ell=1}^L \exp(-i\:g_\ell(x) \: H_l) \: V_\ell(\theta_\ell)
\end{equation} where $H_\ell$s, $g_\ell$s and $\theta_\ell$s are respectively encoding hamiltonians, preprocessing functions, and trainable parameters.
It is known that the quantum models of this family can be expressed as linear models in a given feature space. That is, there exists a \textit{feature space} $\RR^p$ (for some $p\geq1$) and a \textit{feature map} $\phi: \RR^d \longrightarrow \RR^p$ such that the quantum model can be written as
\begin{equation}
    f_Q(x) = \beta_Q^{\top} \, \phi(x) \,,
\end{equation}
with a $\theta$-dependent \textit{weight vector} $\beta_Q^{\top}$.

We are given a \textit{training dataset} of size $M$, consisting of $M$ input points $(x_1, \dots, x_M)$ assumed to have been sampled from some distribution $\mu$ on $\RR^d$, and $M$ scalar targets $(y_1, \dots , y_M)$.
This training inputs in feature space form the \textit{data matrix} $\Phi \in \RR^{M\times p}$ (as $[\Phi]_{ij}=\phi(x_i)_j$), while the outputs $y_i$ form the vector of targets $\by \in \RR^{M\times 1}$.

We assume that the goal is to learn a target function linear in the same feature space
\begin{equation}
    f^*(x) = \beta^{*\top} \, \phi(x) \, 
\end{equation} for some $\beta^* \in \RR^p$. If we relate it to a real world task, there is of course no particular reason for the target function to be expressed in this way. This particular case is however very useful to understand quantum advantage.

During training, the parameters $\theta$ of the quantum circuit are chosen iteratively so as to optimize the \textit{empirical risk} loss
\begin{align}\label{eq:empirical_risk_loss}
    \mathcal{L}(f_Q, f^*) &= \frac{1}{M} \sum_{i=1}^M (f_Q(x_i;\:\theta) - y_i)^2\\
    &= \frac{1}{M}\norm{\Phi\beta_Q(\theta) - \by}^2 \,.
\end{align}

On a theoretical level, how well a particular model $f$ (quantum or classical) \textit{generalizes} to the true solution $f^*$, is captured either by the square of their $L_2$ distance (with respect to distribution $\mu$) 

\begin{equation}\label{eq:true_risk}
    \norm{f - f^*}_\mu^2 := \int_{\mathcal{X}} (f(x) - f^*(x))^2 d\mu(x) \, ,
\end{equation}
or by their $\infty$-distance (assuming $f$ and $f^*$ are bounded)
 \begin{equation}
    \norm{f -f^*}_{\infty} := \sup_{x \in \mathcal{X}}\!\big|f(x) - f^*(x)\big|\,.
\end{equation}

In this work, we will consider that the quantum model has a quantum advantage, if one cannot produce a classical model with similar generalization performance at a cost polynomial in the number of qubits and gates.

A summary of the notations used in the paper can be found in \cref{sec:notations}.

\subsection{Bias of classical linear regression}
\label{subsec:mnls}

In this subsection, we detail the known results \cite{bishop2006pattern, hastie2022surprises} about the solution of the linear regression problem. To minimize the empirical risk loss defined in \autoref{eq:empirical_risk_loss}, one can train the weight vector $\beta$ using a Gradient Descent (GD) method or solve the equivalent Kernel Ridge Regression (KRR) \cite{bishop2006pattern, hofmann2008kernel}. We distinguish two regimes:

\begin{itemize}
    \item The \textbf{underparameterized} regime where the feature space dimension is lower or equal to the number of datapoints: $p \leq M$. In this regime, there is a unique solution, that can be expressed as
\begin{equation}
    \hat{\beta} = (\Phi^\top\Phi)^{-1}\Phi^\top\by \,.
\end{equation}
Furthermore, if the data is sampled such that $\Phi^\top\Phi$ is full rank (which is almost always the case) and there is no noise in the observed targets, the estimator $\hat{\beta}$ is equal to the ground truth $\beta^*$.
    \item The \textbf{overparameterized} regime where the feature space dimension is greater than the number of datapoints: $p > M$. In this case, an infinite number of weight vectors can set the empirical risk to zero. However, the algorithms of GD and KRR will converge towards a specific vector $\beta_{\mathrm{MNLS}}$ called \textbf{minimum norm least square} estimator (MNLS). $\beta_{\mathrm{MNLS}}$ is the vector of minimal norm among the minimizers of the empirical loss, and it is provably unique:
\begin{equation}
    \beta_{\mathrm{MNLS}} = \arg \min \norm{\beta}_2 \; \text{with} \;  \mathcal{L}(\beta^\top \phi, f^{\ast}) = 0 \,
\end{equation}
which can also be written
\begin{equation}%
    \beta_{\text{MNLS}} = \Phi^\top(\Phi\Phi^\top)^{-1}\by
\end{equation}
This behavior is due to the fact that GD and KRR only search a solution in the space spanned by the training datapoints, called the \emph{row space}.
More details can be found in \cref{app:mnls}. This property remains true for gradient descent algorithms with momentum \cite{nakerst2020gradientdescentmomentum}. It is unknown whether it is true for one of the most used optimization algorithm, Adam \cite{kingma2017adammethodstochasticoptimization}.
\end{itemize}

\subsection{Classical models for dequantization}\label{subsec:rff}

Given a quantum model, one can design a \textit{surrogate} model by considering a classical model with the same feature map $\phi$, defining a new linear model:
\begin{equation}
    f_C(x) = \beta_C^{\top} \, \phi(x) \, ,
\end{equation}
with $\beta_C$ a weight vector that is obtained with a classical computer. %
An obvious obstacle for doing so is the fact that the dimension of the feature map $\phi(x)$ is too large to be stored in memory. However, techniques exist to mitigate this problem.

Authors in \cite{rahimi2007random} have introduced the \textit{Random Fourier Features} (RFF) technique to lower computational costs for kernel methods and error bounds were refined in \cite{sutherland2015error, li2019towards}. This method can be generalized for many other cases in classical linear regression \cite{rahimi2008uniform, rahimi2008weighted}, and can be applied to the arbitrary basis quantum models defined with a Hamiltonian encoding with any preprocessing function. It has been shown that approximating the target function can be done by learning a function of the form $\hat{f}(x) = \sum_{k=1}^D \beta_i\phi_k(x)$ with $D << p$ where the functions $\phi(\:\cdot\: ;\omega_k)$ are sampled from $\llbracket 1,\: p \rrbracket$. In this case, one only has to learn a vector of dimension $D$.

Studies such as \cite{landman2022classically, sweke2023potential, sahebi2025dequantization} have shown that Random Features Regression can be used to dequantize
 quantum Fourier models, although limitations exist, particularly for resource-constrained circuits. These works are based on the Fourier feature maps that arise from Hamiltonian encoding \cite{schuld_effect_2021} where the classical data $x$ is carried by an Hamiltonian $H$ through a unitary $e^{ixH}$.

In the following, we state  a generalization to the Random Features Regression technique for any Hamiltonian encoding with any pre-processing $g$ of the classical data $e^{ig(x)H}$ that will be used in the rest of the work, and refer to \cref{app:rff} for more details. 

\begin{restatable}[]{thm}{TheoremRFF}
\label{thm:rff}
    Let  $\phi(x) = [\sqrt{q_1}\:\phi_1(x) \dots \sqrt{q_p}\:\phi_p(x)]^\top$ where $\phi_i(x)$ are basis functions such that $\forall x,\: |\phi_i(x)| \leq 1$ and $q = (q_1, \dots q_p)$ represents a discrete probability distribution, and let $f(x) = \beta^\top \phi(x)$. Let $S$ be a subset of $\llbracket 1, p \rrbracket$ sampled independently with the probability density $q$, with $D = |S|$. Then there exists coefficients $c_1, \dots c_D$ such that  $\hat{f}(x) = \Sum_{k\in S} c_k\phi_k(x)$ satisfies 
    \begin{equation}
    \label{eqn:approx}
        \Vert \hat{f} - f\Vert _\mu = \mathcal{O}\left( \frac{\max_i |\beta_i|\:\norm{\phi_i}_{\mu}\:/\sqrt{q_i}|}{\sqrt{D}} \right) \, .
    \end{equation}
    Applying the above to $\bmnls$ obtained from a kernel matrix $K$ with minimal eigenvalue $\lambda_{\min}(K)$ and target vector $\by$ with $\norm{\by}_\infty \leq 1$ yields coefficients $c_1, \dots c_D$ such that
    \begin{equation}
    \label{eqn:approx_mnls}
        \Vert \hat{f} - \fmnls\Vert _\mu = \mathcal{O}\left(\frac{M\:\max_i\norm{\phi_i}_\mu}{\sqrt{D}\:\lambda_{\min}(K)} \right) \, .
    \end{equation}
\end{restatable}

\begin{figure*}
\begin{center}
\includegraphics[width=0.95\textwidth]{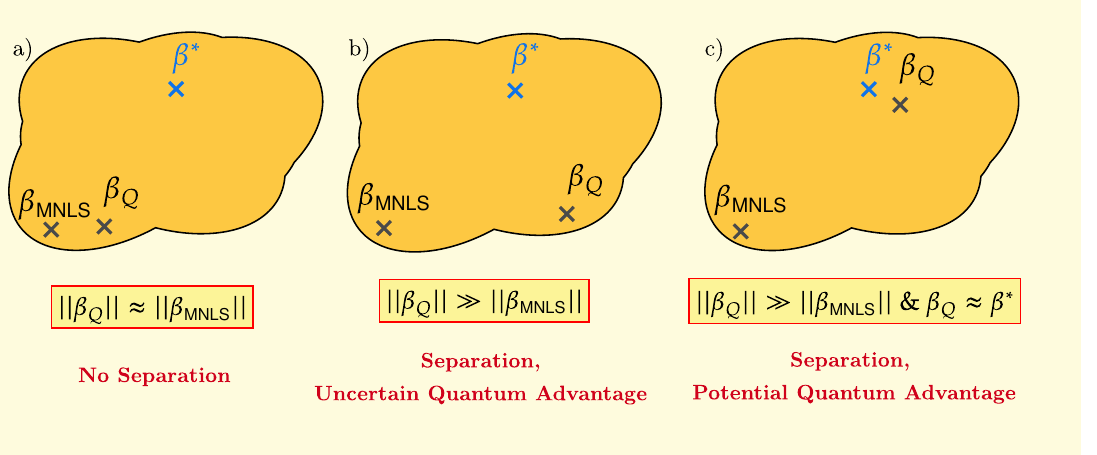}
    \caption{Illustration of the potential quantum advantage. If $\beta_Q$ is close to $\bmnls$ there is no separation between the quantum estimator and the classical one. If $\beta_Q$ and $\bmnls$ are far from each other and far from the ground truth, there is a separation but uncertain quantum advantage. If $\beta_Q$ is closer to the ground truth than $\bmnls$, there is a suggestion of quantum advantage.}
    \label{fig:Quantum_Advantage}
\end{center}
\vspace{-2em}
\end{figure*}

    \section{Bias of quantum models and potential advantage}\label{sec:Bias_Quantum_Advantage}

In our work, we propose a general study of learning through VQCs, that can be applied to the case where the input variable is continuous or discrete. Most proven quantum advantage results in quantum machine learning come from problems where the input data take discrete values \cite{gyurik2023exponential, molteni2024exponential, jerbi2023quantum, liu2021rigorous, jerbi2024shadows}, typically $\{0, 1\}^n$. It is convenient because the problems can be linked to cryptography problems which are known or strongly supposed to be hard to solve classically. However many real world use cases utilize continuous vectors, so it is important to have a better understanding in that domain.

\subsection{Underparameterized regime has few advantages}

First of all, we note that in the underparameterized regime, there is little possibility of solving the linear regression problem in a better way with a quantum computer. The optimal solution to the least square problem has indeed a closed form, and if there is no noise in the data, it is equal to the true weight vector. It means that any other optimization technique will converge towards that optimal solution. Moreover, since we assume that the number of data points is small enough to be handled with a classical computer, the total number of operations in the procedure is still polynomial in the size of the dataset. We do not exclude however an advantage of using a quantum computer to invert the covariance matrix \cite{harrow2009quantum}, or other more modest polynomial advantages \cite{cerezo2022challenges}.%

\subsection{Quantum models can differ from Minimum Norm Least Square}

We consider the overparametrized regime, and we described in \autoref{subsec:mnls} the implicit bias of classical learning algorithms. A classical linear regression trained with gradient descent, or a kernel ridge regression will output a model $f_{\text{MNLS}}(x) = \beta_{\text{MNLS}}^\top\:\phi(x)$, that could be classically approximated in lots of cases if one provides a sampling access to the entries of $\phi(x)$.
Contrary to the classical case, one does not have access directly to the coefficients $\beta_Q$ while tuning a quantum model. One instead optimizes a vector of parameters $\theta$ such that $\beta_Q = \beta_Q(\theta)$ and optimizes the loss function $\mathcal{L}(\theta) = \norm{y - X\beta_Q(\theta)}^2$.

In this case $\beta_Q$ does not remain in the row space (the space spanned by the training datapoints) during the training and does not converge to $\beta_{\text{MNLS}}$. This constitutes a crucial distinction between quantum and classical models. 
If the quantum model would converge to $\fmnls$, it could be approximated with random feature regression techniques. Therefore we suggest that the best usage of quantum computers would not be to reproduce classical linear regressions. The quantum circuit should be used to provide a model $\beta_Q$ such that $\beta_Q \neq \bmnls$. 
It remains to be seen when $\beta_Q$ can converge far from $\beta_{\text{MNLS}}$ or from an approximation of MNLS via random features. 

In practice, we can consider $\norm{\beta_Q} \geq \text{poly}(N)$ in order to have a clear separation. Such examples are developped in \autoref{subsec:Example_quantum_Fourier}. Having a weight vector of large norm will provide a difference with classical models, but a true advantage will be reached if in addition the quantum models is closer to the ground truth than the MNLS. We illustrate these views in \autoref{fig:Quantum_Advantage}.
Since $\fmnls$ is the interpolating model of minimum norm, any quantum interpolating model $f_Q$ must verify $\norm{\beta_Q} \geq \norm{\bmnls}$. A sufficient condition for separation between $\norm{\beta_Q}$ and $\norm{\bmnls}$ would be that $\norm{\beta_Q} \gg \norm{\bmnls}$. We state it in the informal theorem

\begin{restatable}[Informal]{thm}{RelationBetaNormFunctionDist}\label{thm:largeNormBeta}
Let $f_Q$ be an interpolating quantum model, ie $\mathcal{L}(f_Q) = 0$. Then, $f_Q$ has a potential quantum advantage if $\norm{\beta_Q} \gg \norm{\bmnls}$.
\end{restatable}
We refer to \cref{thm:looking-at-norm-betaq-to-assess-potential-quantum-advantage} in \cref{sec:Proof_Error_Model} for more details.

Other works have outlined the differences between quantum and classical linear regression, but none of them mentions the criteria about the norm of the weight vector. The authors in \cite{jerbi2023quantum} study the fact that variational circuits express a different solution than the kernel ridge regression (therefore the MNLS). They point out that there exists functions that are learnable with variational quantum circuits but that require exponentially more resources to learn with quantum kernels. In \cite{you2023analyzing}, the authors analyse the optimization dynamics of QNNs and conclude that they are different from the neural tangent kernel. They study in detail the convergence rate of the respective methods, but do not study the actual solution reached.

In the following we will study the weight vector norms for usual VQC framework with continuous input, but also for the cryptographic examples with discrete inputs described in \cite{liu2021rigorous, gyurik2023exponential}. Our result will emphasize the importance of this criterion to find a quantum advantage. Complements to this discussion can be found in \cref{sec:Proof_Error_Model}.

\subsection{Generalization performances}

In the previous subsections, we only discussed the separation between the classical and quantum models. The separation is a necessary condition for quantum advantage, but the classical model could generalize better than the quantum model.
Without any other information than the norm of the weight vector, we can show better generalization bounds for the MNLS estimator than the quantum model if  $\norm{\beta_Q} \gg \norm{\bmnls}$ (see \cref{app:generalization}, and \cite{mohri2018foundations}).
The advantage of using a quantum model will come therefore from an inductive bias that makes the quantum model more suitable for the data. This would heavily depend on the use case, and it is very difficult to give general statements about such inductive biases.

\section{Examples of separation from classical to quantum models}
\label{sec:examples}

\subsection{Fourier model
}

In this section, we consider Fourier models defined by the Fourier feature map:%
 
\begin{equation}\label{eq:feat-map}
    \phi(x) = \frac{1}{\sqrt{p}}\begin{bmatrix} \cos(\omega^{\top}x) \\ \sin(\omega^{\top}x) \\ \vdots \end{bmatrix}_{\omega \in \Omega}
\end{equation}
with $\Omega \subset \mathbb{Z}^d$, and $p=|\Omega|$. We assume that the spectrum is only composed of vectors of integers, and we can consider that $\forall\: \omega \in \Omega, -\omega \notin \Omega$. We also consider the input vector $x$ to be uniformly distributed in $[0, 2\pi]^d$.
This case is very important in the quantum machine learning literature \cite{schuld2021effect, schreiber2022classical, sweke2023potential, peters2022generalization}, and could help us understand what quantum circuit design needs to be done in order to do variational circuit learning. 

It has been explained in \autoref{subsec:rff} that the approximability of the MNLS estimator depends on the eigenvalues of the empirical kernel matrix. We prove in \cref{app:Concentration_Eigen_K_Fourier_integer} that the smallest eigenvalue of the kernel matrix is with high probability a constant. It enables us to state that the MNLS associated to a Fourier model can be easily approximable, and that the norm of $\norm{\bmnls}$ is bounded by the number of datapoints. This is stated in the following theorem.

\begin{restatable}[]{thm}{WeightVectoretwodesign}\label{thm:WeightVectoretwodesignInformal} (Informal)
Let us consider a Fourier model with a spectrum $\Omega \subset \mathbb{Z}^d$. %
Then, the associated MNLS estimator has a norm scaling like $\norm{\bmnls} = \mathcal{O}(M)$.
\end{restatable}

In the case of Fourier models, \autoref{eqn:approx_mnls} from \autoref{thm:rff} can be rewritten, such as: 

\begin{equation}\label{eqn:approx_mnls}
    \Vert \hat{f} - \fmnls\Vert _\mu = \mathcal{O}\left( \frac{M}{\sqrt{D}} \right) \, .
\end{equation}

It is the most favorable case to apply random feature regression, it is then enough to have a number of random features polynomial in $M$ to approximate the MNLS estimator.

\subsection{Simple quantum Fourier model}\label{subsec:Example_quantum_Fourier}
\vspace{-.5em}

In this section, we detail an example of quantum Fourier model that exhibit the separation mentionned in the previous section, ie $\norm{\beta_Q} \gg \norm{\bmnls}$. More details can be found in \cref{app:weingarten}

We consider a circuit with a diagonal encoding layer $S(x)$ applied to the $|+\rangle^n$ state followed by a trainable unitary $V$ and an observable $O$ such that $\text{Tr}(O) = 0$. We illustrate this example in \autoref{fig:Potential_Advantage_Framework}. 
The quantum model can then be written
\begin{equation}
    f_Q(x) = \text{Tr}(O\: VS(x)(|+\rangle\langle +|^n)S(x)^\dagger V^\dagger)
\end{equation}

\begin{figure}[h!]
    \centering
    \includegraphics[width=0.5\textwidth]{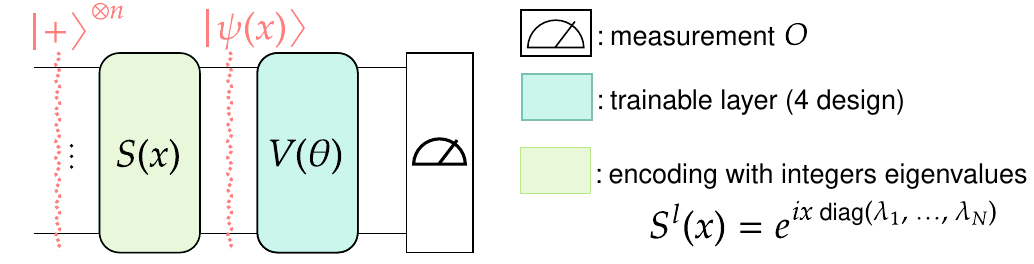}
    \caption{Parameterized quantum models considered: hamiltonian encoding with no integer eigenvalues.}
    \label{fig:Potential_Advantage_Framework}
\end{figure}

The spectrum only depends on the encoding unitary $S(x)$. We consider two types of encodings
among many possibilities that are detailed in \cite{peters2022generalization}:

\begin{itemize}
    \item \textbf{The ternary encoding}
    
    $S(x) = \bigotimes_{k=0}^{n-1} RZ_{k}(x\:3^k/2)$
    where $RZ_{k}$ denotes a $Z$ rotation applied to the qubit $k$. The spectrum produced by this encoding is the interval $\llbracket 0, 3^n-1\rrbracket.$ It is the spectrum with the largest size one can produce with one layer of single qubit gates \cite{shin2023exponential}.
    
    \item \textbf{The Golomb encoding}
    $S(x) = \exp(-i \displaystyle\frac{x}{2} R_G)$ where $R_G$ is a Golomb ruler \cite{piccard1939ensembles}. The resulting spectrum are all the integers in the set $\llbracket 0, N(N-1)/2\rrbracket.$ Such an encoding is not known to be realizable in polynomial time on a quantum computer, so is of little practical use, but it is an interesting edge case of our results.
\end{itemize}

We do not detail here the multiplicity of the spectrum, it can be found in \cref{app:weingarten}.
In this setting, we have the following result:

\begin{restatable}[]{thm}{WeightVectoretwodesign}\label{thm:WeightVectoretwodesignInformal} (Informal)

Let us consider a Fourier model as defined above, with a unitary $V$ sampled on a 4-design.
With high probability over the distribution of the dataset,
$\norm{\bmnls}^2 \leq M^2$.
For the Golomb encoding, $\norm{\beta_Q}^2 \sim 2^n$ with high probability over the distribution of $V$. For the ternary encoding, $\norm{\beta_Q}^2 \sim (3/2)^n$ with high probability over the distribution of $V$.

\end{restatable}

We obtain this result by integrating order 2 and 4 moments of the Haar measure \cite{collins2006integration, weingarten4}. 
This result shows that we can have a separation between quantum and classical models in the special cases mentioned. Indeed, we have that $\norm{\beta_Q}\gg \norm{\bmnls}$ if we assume that $M = O(\text{poly}(n))$. This result is based on a strong randomness hypothesis, but can be used as a guideline to design quantum models that can present large weight vectors and avoid RF techniques. \autoref{thm:WeightVectoretwodesignInformal} can be extended to any exponentially large encoding spectrum. In \autoref{subsec:Verifying_Weight_Vector}, we offer a method to derive the weight vector norm, allowing one to verify that a quantum model can avoid such dequantization techniques.

However, these models are not suitable to be used in a practical case, because
 considering that the trainable unitary is drawn from a 2-design implies the model concentration and vanishing gradient phenomenon called Barren Plateau  \cite{Holmes2022, mcclean2018barren}. This point is discussed in \cref{subsec:LinkConcentration}

\subsection{Re-uploading Fourier models }\label{subsec:ReUploadingModel}

\vspace{-.5em}

We now consider the case of re-uploading model, where the quantum model $f_Q(x,\theta) = \bra{0} U(x,\theta)^\dagger O U(x,\theta) \ket{0}$, is such that the circuit unitary is composed of an \textit{encoding} layer surrounded by two \textit{trainable} layers of the form: 
\begin{equation}\label{eq:circuit_ansatz}
    U(x,\theta) = V^{2}(\theta) S(x) V^{1}(\theta) \, \textrm{,}
\end{equation}
with $V^1(\theta)$ and $V^2(\theta)$ formed by trainable gates depending on the parameter vector $\theta$, which is optimized during training whereas $S(x)$ only depends on input data values. We consider the Hamiltonian encoding strategy where the classical input components are encoded as the time evolution of some Hamiltonians $S(x)= \prod_{k=1}^D e^{-ix_kH^{(k)}}$. As explained earlier, the quantum model can be written as a Fourier Series where its spectrum $\Omega$ which depends on the eigenvalues of the encoding Hamiltonians. %

In \cite{mhiri2024constrained}, the authors showed that the variance of the Fourier coefficients $c_\omega(\Theta)$ depends on the \textit{redundancy} $|R(\omega)|$ of their corresponding frequencies $\omega$. We use the results from this work to give the following results. In this work, we are investigating the difference of learning behaviour between the quantum models and the minimum norm least square (MNLS) estimator. Because the norm of the MNLS estimator is bounded, we show how close it is to $\beta_Q$ by considering an upperbound on its norm. First, we consider the case where the trainable layers, $V^1(\theta)$ and $V^2(\theta)$, described 2-design over the special unitary group. 

\begin{restatable}[]{thm}{ExpBetaQReuploadingtwodesign}\label{thm:Exp_BetaQ_Reuploading_2design}
  (Informal) Consider a single layered quantum re-uploading model with an observable $O$ such that $\text{Tr}(O) = 0$, and $\norm{O}^2_2 = N$. Then, $\norm{\beta_Q }_2 \sim \frac{p}{N}$ with $p$ the number of features, and $N$ the number of distinct eigenvalues in the encoding layer.
\end{restatable}

Therefore, we have that the norm of $\beta_Q$ can be  growing inverse exponential with respect to the number of qubits for  any $p$ polynomial or smaller with respect to the number of qubits, while the case $p \sim N$ may offer a potential advantage. As in the previous example, considering $p \sim N^2$ leads to a clear separation where $\norm{\beta_Q} \gg \norm{\bmnls}$. In \cite{mhiri2024constrained}, the authors offer a bound on the variance of Fourier coefficients according to the monomial distance $\varepsilon$ of each trainable layer unitary matrix to a 2 design. Similarly, we provide a bound on the variance of the weight vector norm in \autoref{app:Reuploading_Proof}, along with more detailed theorems and their corresponding proofs.

Under the hypothesis that the quantum circuit solution minimizes the empirical risk, the $\ell_2$-norm of the quantum circuit weight vector is lower bounded by $\norm{\beta_{\text{MNLS}}}_2$. Those two results can be seen as contradictory, but it simply means that if the trainable layers are close to a 2-design, it could be hard for the quantum circuit to reach a solution that minimizes the empirical risk. %

\subsection{Verifying the weight vector condition in practice}\label{subsec:Verifying_Weight_Vector}
\vspace{-.5em}

Previously, we offered guidelines to design circuits with continuous inputs and weight vector norm highe than the norm of the MNLS estimator, in order to avoid RF approximation techniques. However, \autoref{thm:WeightVectoretwodesignInformal} and \autoref{thm:Exp_BetaQ_Reuploading_2design} are not sufficient to design models that avoid such surrogate techniques, as they are average results based on hypothesis than lead to model concentration and trainability issues. We propose a simple method to verify that a non concentrated quantum model avoids RF approximation techniques. The method relies on the following proposition.

\begin{restatable}[Estimating weight vector norm of quantum models]{prop}{VerifyWeightVector}\label{prop:Verify_weight_vector}
  For a quantum Fourier model $f_Q(x,\theta)$ that has integer frequencies $\Omega \subset \mathbb{Z}$, the following formula holds:
  \begin{equation}\label{eq:beta-norm-estimator}
      \norm{\beta}_2 = \sqrt{p} \Big(2\,\EE_x\big[f_Q(x,\theta)^{\,2}\big] - \EE_x\big[f_Q(x,\theta)\big]^{\,2}\Big)^{1/2},
  \end{equation}
  where the expectation values on $x\in \mathbb{R}^d$ are taken with respect to the uniform probability measure over $[0,2\pi]^d$.
\end{restatable}

 For fixed trainable parameters $\theta$, one can estimate the weight vector norm $\norm{\beta}_2$ by evaluating the quantum model $f_Q(x,\theta)$ at many uniformly-random data points $x \in [0,2\pi]^d$, and then post-processing via \cref{eq:beta-norm-estimator} (with empirical expectation values in place of $\EE_x[\cdot]$'s).

Assuming that the considered model does not suffer from concentration (see \cref{subsec:LinkConcentration}), this method can be used in practice to verify the necessary condition of weight vector norm separation offered in this work.
For concentrated models, such separations (such as those exhibited in \cref{subsec:Example_quantum_Fourier} and \cref{subsec:ReUploadingModel}) cannot be measured precisely using this method (it would require exponentially many shots).
The proof of \Cref{prop:Verify_weight_vector} follows directly from the orthogonality of sine and cosine functions with different integer frequencies (see the beginning of the proof of \autoref{thm:Example_Perfect_Fct},  in \autoref{app:non-concnetration} for details).
Note that this assumption of integer frequencies required in \Cref{prop:Verify_weight_vector} is satisfied by most common encoding strategies, including Pauli gates and the other examples of encodings mentioned earlier.

\subsection{Discrete Inputs VQCs}
\vspace{-.5em}

Learning tasks involving discrete inputs and cryptography primitives have little real world applicability, but they remain interesting for theoretical purposes.
They are indeed the only ones where a quantum advantage has been shown rigorously. \cite{gyurik2023exponential, molteni2024exponential, jerbi2023quantum, liu2021rigorous, jerbi2024shadows}. In this section, we detail an example of a quantum model that provably cannot be dequantized, and explain how it fits into our general theory. More details can be found in \cref{app:discreteLog}.

Let the discrete logarithm unitary be defined as
\begin{equation}
    U_{\text{DLP}}: |i\rangle \longmapsto |\log_g i + 1 \rangle
\end{equation}
where $g$ is a prime number in $\llbracket 0, N-1\rrbracket$.

Let $|\psi (x)\rangle = \bigotimes_{i=1}^n RY(x_i) |0^n\rangle$ and 
\begin{equation}
    f_{\text{DLP}}(x) = \text{Tr}(U_{\text{DLP}}^{\dagger}Z_n U_{\text{DLP}}|\psi(x)\rangle\langle\psi(x)|)
\end{equation}

$U_{\text{DLP}}^{\dagger}Z_n U_{\text{DLP}}$ is a hermitian diagonal matrix and the coefficients can be written as $(U_{\text{DLP}}^{\dagger}Z_n U_{\text{DLP}})_{ii} = (-1)^{b_n(\log i + 1)}$. $b_n(j)$ is the n-th bit of the binary description of $j$. $f_{\text{DLP}}$ can be rewritten as 

\begin{align}
f_{\text{DLP}}(x) %
&= \sum_{y \in \{0, 1\}^n} \beta_y \phi_y(x)
\end{align}
where
\begin{equation}
    \phi_y(x) = \frac{1}{2^d} \prod_{i=0}^{d-1} (1 + (-1)^{y_i}\cos (x_i)), \quad y \in \{0, \:1\}^d
\end{equation}

and $\beta_y = (-1)^{b_n(\log i + 1)}$.

The bounds on the efficiency of Random Feature Regression become exponential in the number of qubits, since for all $y$

\begin{align}
    |b_y \: 2^n\norm{|\phi_y(x)}_{\mu} = \sqrt{2}^n\,.
\end{align}

Therefore it cannot be shown that $f_{\text{DLP}}$ can be learned with naive RF regression. It is consistent with the fact that it cannot be efficiently approximated because of the hardness of the discrete logarithm.
This hardness is not a consequence of what we have shown here, but is a separate statement \cite{liu2021rigorous}.

\vspace{-.5em}

\section{Discussion}\label{sec:Discussion_Limitations}

\vspace{-.5em}

\subsection{Avoiding concentration issues}\label{subsec:LinkConcentration}
\vspace{-.5em}

Concentration phenomenon of parameterized quantum circuits have been studied a lot in the literature. 
We say that a function is concentrated if the variance $\VV_x[f(x)]$  is inverse exponential with respect to the number of qubits. Typically, a quantum model  $f$ is considered concentrated if $\VV[f]\leq 1/\text{poly}(N)$. It is equivalent to the Barren Plateau phenomenon \cite{mcclean2018barren}, where the gradient of the loss function is exponentially close to 0.

The quantum model $f_Q$ can only be estimated by taking an average of $N_{\text{shots}}$ measurements with a precision of $ 1/\sqrt{N_{\text{shots}}}$. Thus if $f$ is concentrated, it would take an exponential amount of shots to evaluate it reliably. Therefore it would not be useful as a model. For the Fourier model, the variance of the function is given by the norm of $\beta$, $\VV[f] = \norm{\beta}^2/p$.
For the weight vectors of the proposed random quantum circuits in \autoref{sec:examples}, we have that $\VV[f]$ is of the order of $1/2^n$ which is concentrated.
A Fourier model with a weight vector norm $\norm{\beta}^2/p \geq \frac{1}{\text{poly}(d)}$ would then be non concentrated and non dequantizable with random features regression.

Constructing such non concentrated quantum models has proven to be a challenge. Recently, the commmunity investigated links between concentration and classical simulability \cite{cerezo2023does}. It has been conjectured that quantum models that do not suffer from concentration can be simulated efficiently with classical computers.

Concentration can be avoided in the case of discrete data input where $x \in \{0, 1\}^n$. \cite{cerezo2023does} gives some examples, and the discrete log function $f_{\text{DLP}}$ is not concentrated because $\EE_{x \in \{0, 1\}^n}[f_{\text{DLP}}(x)^2] - \EE_x[f_{\text{DLP}}(x)]^2 = 1/2 - 1/4 = 1/4$ since for half the inputs, $f_{\text{DLP}}(x)=1$ and $0$ for the other half.

We are interested to construct similar examples, but using data from a continuous distribution, ie $x \sim \mathcal{U}([0, 2\pi]^d)$.
We would like to find functions such that $\norm{\beta_Q}^2\geq p$ , and we wonder if it can be compatible with the fact that $f$ should be bounded independently of $p$, ie $|f(x)|\leq 1$ for all $x$. The fact that $f$ should be bounded comes from the fact that it is the expectation value of an observable. 

In the following, we propose a special family of Fourier model such that the norm of the weight vector is large (thus far from MNLS), and that is not concentrated. In addition, we show that this function is bounded. If one could find a quantum circuit architecture that realizes a function from this family, the conditions for a potential quantum advantage presented previously would be satisfied. An example is found in \autoref{thm:Example_Perfect_Fct}, and the proof can be found in \autoref{app:non-concnetration}.

\begin{restatable}[]{thm}{ThmFunctionFarMNLSNoConcentration}
\label{thm:Example_Perfect_Fct}
    Let $\Omega$ a subset of $\llbracket -L, L \rrbracket ^d$, where $L$ is an integer. We consider the following function $f : \mathbb{R}^d \longrightarrow \mathbb{R}$
    \begin{equation}
        f(x) = \frac{1}{\sqrt{p}}\sum_{\omega \in \Omega}  (\beta_{\omega, \cos}\cos(\omega^\top x) + \beta_{\omega, \sin}\sin(\omega^\top x))
    \end{equation}
    with $p = |\Omega|$, and $\beta_{\omega,\cos}$, $\beta_{\omega,\sin}$ are all iid uniform random variables in the interval $[-\sigma, \sigma]$ with $\sigma = \Theta(1/(d\:(\log d + \log L)))$
    We have the following properties:
    \begin{enumerate}
        \item $\big|\Vert \beta \Vert^2 - \frac{2}{3}p\sigma^2\big| \leq \sigma^2 \sqrt{p\log (2/\delta)}$ with probability at least $1-\delta$ over the choice of $\beta$.
     \item $\VV_x[f(x)] \geq \frac{2}{3}\sigma^2 - \frac{\sigma^2}{\sqrt{p}} \sqrt{\log (2/\delta)}$ with probability at least $1-\delta$ over the choice of $\beta$.
     \item $\forall x \in \RR^d, \:|f(x)| \leq 1$ with high probability over the choice of $\beta$.
    \end{enumerate}
\end{restatable}

In the above theorem, (1) shows that $\norm{\beta}^2$ is of the order of $p \sigma^2$ therefore of potentially higher norm than $\beta_{\text{MNLS}}$ (which scales like $M$), (2) shows that $f$ is not concentrated, and (3) shows that $f$ is bounded by a constant, which leaves open the amenability to realize it as an quantum expectation value of an observable $O$ with $\norm{O}_{\infty}$ bounded by a constant, which is a property of commonplace quantum observables. This theorem gives a function that is not impossible to achieve from a VQC, far from the corresponding classical model, and not concentrated. However, one needs to find a quantum circuit capable of implementing such a function.

\vspace{-1em}

\subsection{Open Questions}\label{subsec:Limitations}
\vspace{-1em}

In the previous sections, we showed how quantum models can often converge to a solution close to the MNLS estimator, and how to offer potential quantum advantage over RF surrogates using arguments on the norm of the weight vector. One could design a VQC based on those indications to create a model far from its classical counterpart. In this Section, we offer open questions and suggestions for future work.

First, we focused our analysis on usual classical gradient descent and KRR, giving rise to the bias of converging towards the MNLS estimator. While our results can be extended for momentum based gradient descent \cite{nakerst2020gradientdescentmomentum} (see \autoref{app:mnls}), we did not explore other classical learning algorithms that may not converge to the same solution.

Our work shows how to obtain separation between quantum and classical learning models through the study of their weight vectors. This condition could be applied to VQCs with discrete or continuous input variables. However, finding a VQC with continuous variable input and without concentration issues can be challenging, while possible as explained in \autoref{thm:Example_Perfect_Fct}. Finding such VQCs or how to use ones with discrete input variables for useful learning problems must be tackled in future works. 

One could envision a change of the feature map, as the characteristics of the MNLS solution may then differ. Our study would need to be refined but could be adapted to any feature map. However, an intuitive case where one expects to find a quantum advantage would be when the individual components of the feature map are functions that are easy to compute on a quantum computer, but hard to do so on a classical one. Therefore even trying to train a classical surrogate wouldn't be possible. For instance, one can create a feature map inspired by cryptography \cite{Shor_1997,liu2021rigorous}, or create feature maps involving the ground state of data dependent Hamiltonians \cite{umeano2024ground}.

Finally, while our work focus on avoiding RF based surrogate techniques, a practitioner would need to study more in details the training of the quantum model. In addition with concentration issues mentioned in this paper, other training issues can occur including local minima.

\vspace{-1em}
\section{Conclusion}
\vspace{-1em}

Using variational quantum circuits (VQCs) for machine learning was originally motivated by the fact that they realize linear models in high dimensions. However, classical dequantization methods have been proposed to approximate certain forms of these quantum models.

We point out that, for the same feature map, quantum and classical models do not converge towards the same solution.  Classical linear regression trained with gradient descent or via the associated kernel ridge regression (KRR) converges towards the minimum norm least square estimator (MNLS). In the overparametrized regime (feature space dimension higher than size of training dataset) however, an infinite number of solutions fit the training data, and quantum models generally do not converge to the MNLS, due to the special parametrization of their weight vector.
It could be that this inductive bias of quantum models leads to solutions that generalize better to the ground truth function.
Besides, since the MNLS solution can be anyways approximated to some degree via regression on a polynomial number of random features, the use of a quantum computer to obtain it may be unecessary. However, since the quality of this approximation is also influenced by the spectrum of the empirical kernel matrix (which depends on the exact feature map and data distribution), we do not rule out a case where quantum computers could provide an advantage by evaluating the MNLS.

We investigate in greater details the case of the Fourier feature map with integer coefficients. In this case, the condition number of the kernel matrix is constant in high probability whenever there is polynomially more features than datapoints, so it is the most favorable case to apply random features regression. We show examples of quantum circuits that implement the same feature map and whose weight vector has a norm much bigger than the norm of the MNLS. Therefore there could exist a separation between quantum and classical models in this cases.

Unfortunately, the proposed quantum models are highly concentrated, which makes them unusable in practice. The concentration of a linear function for integer Fourier features with uniform distribution also directly depends on the norm of the weight vector. We asked whether it was possible to simultaneously have a  weight vector norm scaling like the number of features and the function to be bounded. We show that it is possible and we exhibit a function with these properties. The quantum circuit realizing these models remains to be found.

\section{Acknowledgments}

The authors warmly thank Hela Mhiri, and Elham Kashefi for the fruitful discussions. All the work was conducted at Sorbonne University and University of Edingburgh. ST was supported by Pasqal. LM was supported by the Naval Group Centre of Excellence for Information Human factors and Signature Management (CEMIS). EM is supported by the grant ANR-22-PNCQ-0002.

\section{Author contributions}
S.T and J.L conceived the project. S.T, L.M, E.Z.M, J.L contributed to the theoretical developments. All authors contributed to the technical discussions and writing of the manuscript.

\bibliographystyle{IEEEtran}
\bibliography{refs}

\clearpage

\onecolumngrid
\appendix

\tableofcontents

\newpage
\section{Table of notations}
\label{sec:notations}

\begin{table*}[h!]
\begin{tabular}{|c|c|c|c|}
\hline
Notation                  & Object                                       & Notation   & Object                                 \\ \hline
$x$                       & Input vector                                 & $M$        & Size of the dataset                    \\
$d$                       & Input dimension                              & $\bPhi$    & Data matrix of shape $(M, p)$          \\
$\phi(x)$                 & Feature map                                  & $\by$      & Vector of targets of shape $(M, 1)$    \\
$p$                       & Given context: dimension of $\phi(x)$, size of $\Omega$  & $D$        & Number of random features              \\
$f_Q(x; \theta), \beta_Q$ & Quantum model and its weight vector          & $\Omega$   & Spectrum of the quantum model          \\
$\theta$                  & Parameters of the Quantum circuit            & $\Omega^*$ & Spectrum of the quantum model without 0\\
$f^*, \beta^*$            & Target function and its weight vector        & $\Omega_+$ & Subset of $\Omega^*$ such that $\forall \omega \in \Omega_+, -\omega \in \Omega\backslash\Omega_+$\\
$\fmnls, \bmnls$          & Minimum norm estimator and its weight vector &  $n$       & Number of qubits\\
$K$                          &        Kernel matrix                                      & $N$        & $N=2^n$, size of the Hilbert space     \\ \hline
\end{tabular}
\caption{Main notations used in the paper.}
\label{tab:notations}
\end{table*}

\section{Minimum norm least square estimator}
\label{app:mnls}

In this section, we give more details about the overparameterized regime of classical linear regression. In this case, an infinite number of weight vectors can set the empirical risk is zero. However, the algorithms of GD and KRR will converge towards a specific vector $\beta_{\mathrm{MNLS}}$ called \textbf{minimum norm least square} estimator (MNLS) \cite{hastie2022surprises}. $\beta_{\mathrm{MNLS}}$ is the vector of minimal norm among the minimizers of the empirical loss, and it is provably unique:
\begin{equation}
\label{eq:bmnls}
    \beta_{\mathrm{MNLS}} = \arg \min \norm{\beta}_2 \; \text{with} \;  \mathcal{L}(\beta^\top \phi, f^{\ast}) = 0 \,
\end{equation}
which can also be written
\begin{equation}
    \beta_{\text{MNLS}} = \Phi^\top(\Phi\Phi^\top)^{-1}\by
\end{equation}
This result is formalized in \autoref{theorem:ConvergenceMNLS}. It is not a new result, but we give a proof for completeness.
The key reason for such a property is that when performing GD, each iterate of $\beta$ stays in the subspace $V = \text{span}(\phi(x_1), \dots \phi(x_M))$, the \textit{row space} of the data. Each vector of $\RR^p$ can be decomposed as $\beta = \beta_V + \beta_{V^\perp}$ where $\beta_{V^\perp} \in V^\perp$, the orthogonal of $V$ in $\RR^p$, such that $\Phi\beta_{V^\perp} = 0$.

Each vector $\hat{\beta}$ such that $\mathcal{L}(\hat{\beta}) = 0$ can then be written 
\begin{equation}
    \hat{\beta} = \bmnls + u \, ,
\end{equation}
where $u \in V^\perp$, and $\Phi\bmnls = \by, \: \Phi u = 0, \: \Vert \hat{\beta}\Vert ^2 = \Vert \beta_{\text{MNLS}}\Vert ^2 + \Vert u\Vert ^2$.
Since the iterates of GD stay in $V$, then $u=0$ all along.

\begin{restatable}[From \cite{hastie2022surprises}]{thm}{ConvergenceMNLS}
\label{theorem:ConvergenceMNLS}
    Let $\beta_0 = 0$ the initialization of a gradient descent algorithm. Let the following iterations be defined by 
    \begin{equation}
    \label{eqn:gd}
        \beta_{k+1} = \beta_k + \gamma\:\Phi^\top (\by - \Phi\beta_k) \, ,
    \end{equation}
    with $\gamma$ the learning rate such that $0 \leq \gamma \leq 1/\lambda_{\max}(\Phi^\top X)$ where $\lambda_{\max}(\Phi^\top \Phi)$ is the largest eigenvalue of $\Phi^\top \Phi$. Then: 
    \begin{itemize}
        \item $\beta_k$ converges towards the minimum norm least square estimator $\bmnls$ defined in \autoref{eq:bmnls}.
        \item $\beta_{\rm MNLS} = \Phi^\top(\Phi\Phi^\top)^{-1}\by$.
    \end{itemize}  
\end{restatable}

\begin{proof}
     The algorithm converges to a minimizer of $\norm{y - X\beta}^2$ noted $\hat{\beta}$ for the choice of the step size. Furthermore, at each iteration iteration of the algorithm, $\beta_k$ lies in the row space of $X$. Then $\hat{\beta}$ also lies in the row space of $X$, and we will show that it is necessary the minimum norm least square.
     
     If $(X^\top X)$ is low rank, then the minimum least square is unique as shown previously.
     
     Otherwise, if $(XX^\top)$ is full rank, let us denote $\beta^* = X^\top \alpha^\star = X^\top (XX^\top)^{-1}y$. $\beta^*$ is a minimizer of the least square loss and is in the row space of $X$. We also have that $X\beta = y$. We will show that there are no other minimizer of the least square loss on the row space of $X$ and that $\beta^*$ is the minimum norm least square.
     Let $\beta = X^\top \alpha$ another minimizer of the least square loss in the row space of $X$. We then necessarily have $X\beta = y = X\beta^*$. Then $X(\beta - \beta^*) = XX^\top (\alpha - \alpha^*) = 0$, and $\alpha = \alpha^*$ since $XX^\top$ is full rank.
     
     Let $\beta = X^\top\alpha + v$ be a minimizer of the least square loss where $v$ is in the orthogonal of the row space of $X$. We then have 
     \begin{align*}
         X\beta &= X\beta^*\\
         XX^\top \alpha + Xv &= XX^\top \alpha^*\\
         XX^\top(\alpha^* - \alpha) &= Xv
     \end{align*}
     Let us compute $\norm{X^\top\alpha}^2 - \norm{X^\top\alpha^*}^2$.
     \begin{align}
         \norm{X^\top\alpha}^2 - \norm{X^\top\alpha^*}^2 &= (X^\top \alpha)^\top (X^\top\alpha - X^\top \alpha^*) + (X^\top\alpha - X^\top \alpha^*)^\top (X^\top\alpha^*)\\
         &= (X^\top\alpha + X^\top\alpha^*)^\top (X^\top\alpha - X^\top\alpha^*)\\
         &= (\alpha + \alpha^*)^\top XX^\top(\alpha - \alpha^*)\\
         &= -(\alpha + \alpha^*)^\top Xv\\
         &= - (X^\top\alpha + X\top\alpha^*)^\top v\\
         &= 0
     \end{align}
     where at the last equality we use the fact that $v$ is orthogonal to the row space of $X$.
     Finally we have $\norm{\beta}^2 - \norm{\beta^*}^2 = \norm{X^\top\alpha}^2 + \norm{v}^2 - \norm{X^\top\alpha^*}^2 = \norm{v}^2 > 0$ if $v$ is non zero.
     Then the minimum norm least square is the unique least square estimator to be in the row space of $X$.
\end{proof}

This argument remain true for momentum based gradient descent \cite{nakerst2020gradientdescentmomentum}. In this case, the iterations are defined as
\begin{equation}
    \begin{split}
        m_{t} &= g\: m_{t-1} - \gamma\nabla_{\beta_t}\mathcal{L} \:\: ; \:\: m_0 = 0\\
        \beta_{t+1} &= m_t +\beta_t \:\: ; \:\: \beta_0 = 0
    \end{split}
\end{equation}
where $m$ is an intermediate variable, called momentum, and $g>0$ is a hyperparameter. Since $\nabla_{\beta_t}\mathcal{L}$ is in the row space of the data, the momentum is also in the row space at each iteration, as well as the values of the parameters $\beta_t$.

\section{Random Feature regression}
\label{app:rff}

\begin{figure*}
\begin{center}
\includegraphics[width=0.95\textwidth]{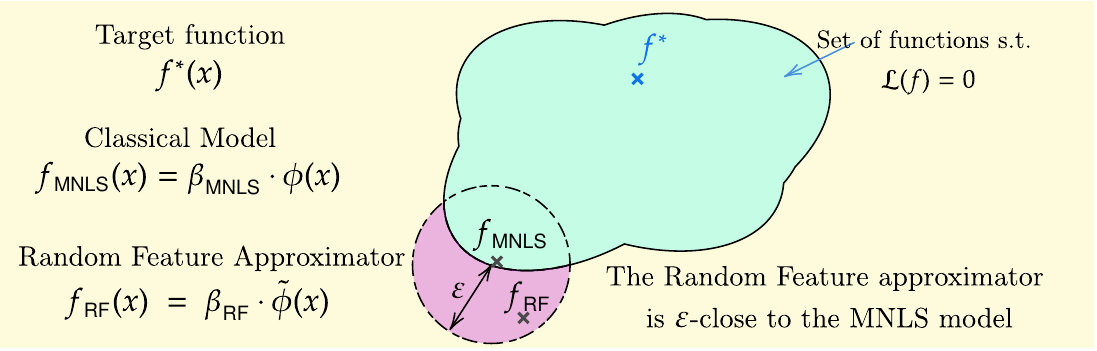}
    \caption{The random feature approximator approximates the MNLS of a feature map, given suitable behavior of the kernel matrix, which depends on the feature map and the data distribution. (see \cref{thm:rff})}
    \label{fig:RFF_MNLS}
\end{center}
\vspace{-2em}
\end{figure*}

We illustrate the relation between the target function, the MNLS function, and the random feature approximator in \autoref{fig:RFF_MNLS}.

We recall \autoref{thm:rff}:
\TheoremRFF*

\begin{proof}
 The \autoref{eqn:approx} can be proven with the same technique as  theorem 3.1 from \cite{rahimi2008uniform}.
    We will use \autoref{lemma:approxHilbert} of \cite{rahimi2008uniform} (reminded in below).\:
    Let $S$ be a subset of $\llbracket 1,\: p\rrbracket$ of cardinal $D$ sampled independently with the distribution $q$. Let $\hat{f} = \frac{1}{D}\sum_{k \in S} \frac{\beta_k}{\sqrt{q_k}}\phi_k$. Then $\EE[\hat{f}] = f$ and $\norm{\frac{\beta_k}{\sqrt{q_k}}\phi_k}_\mu = |\frac{\beta_k}{\sqrt{q_k}}| \norm{\phi_k}_\mu\leq \max_i|\beta_i\norm{|\phi_i}_\mu/\sqrt{q_i}$
    
    To derive \autoref{eqn:approx_mnls}, we note that $\beta_i = \sqrt{q_i}\;\boldsymbol{k}^{\top}K^{-1}\by$ where $\boldsymbol{k} = [\phi_i(x_1) \dots \phi_i(x_M)]^{\top}$. We have then
    \begin{align}
        |\beta_i/\sqrt{q_i}| &\leq |\boldsymbol{k}^{\top}K^{-1}\by|\\
        &\leq \norm{\boldsymbol{k}}\norm{K^{-1}\by}\\
        &\leq \frac{M}{\lambda_{\min}(K)}
    \end{align}
    The second equation is obtained using Cauchy Schwarz inequality, and the last one by using the facts that $\norm{\boldsymbol{k}} \leq \sqrt{M}$ because $\norm{\boldsymbol{k}}_\infty \leq 1$, and $\norm{\by} \leq \sqrt{M}$ because $\norm{\by}_\infty \leq 1$
\end{proof}

\begin{lemma}
\label{lemma:approxHilbert}
Lemma 4 of \cite{rahimi2008weighted}
    Let $X = \{x_1, \dots, x_D\}$ be iid random variables in a ball of radius $R$ centered around the origin in a Hilbert space. Let $\bar{X} = \frac{1}{D}\sum_k x_k$ their average. Then for any $\delta > 0$, with probability at least $1-\delta$, 
    \begin{equation}
        \norm{\bar{X} - \EE[\bar{X}]} \leq \frac{R}{\sqrt{D}}\big(1 + \sqrt{2\log \frac{1}{\delta}}\big)
    \end{equation}
\end{lemma}

\section{Generalization}
\label{app:generalization}

 Some usual generalization bounds from classical machine learning theory can link generalization to the norm of the weight vector. To state one example explicitly, the direct application of Thm. 11.3 (using also Thm. 6.12) in \cite{mohri2018foundations} gives that for $M$ random iid samples of data points $x_i$ (from probability measure $\mu$ on $\mathcal{X}$), it holds with probability at least $1-\delta$ over the choice of the dataset that:
\begin{align}\label{eq:genbound}
R(\beta) &\leq 2 \cdot 4(\norm{\beta}_2 + \norm{\beta^\ast}_2) \cdot \frac{\norm{\beta}_2}{\sqrt{M}} \,+\, 3 (\norm{\beta}_2 + \norm{\beta^\ast}_2)^2 \sqrt{\frac{\log(\frac{2}{\delta})}{2M}}\
\end{align}
with $R(\beta) := \norm{f_\beta - f_{\beta^\ast}}_\mu^2 = \int_{\RR^d} (f_\beta(x) - f_{\beta^\ast}(x))^2 d\mu(x)$, and we assumed zero training error for $\beta$.
So indeed, the lower the weight vector norm $\norm{\beta}_2$ is, the more it is guaranteed (via Eq.~\eqref{eq:genbound}) that the generalization error $R(\beta)$ is low.
However these generalizations bound may generally be loose, i.e. a model (quantum, say) different from the MNLS model, hence with a higher weight vector norm $\norm{\beta}_2$, does not necessarily have a higher generalization error $R(\beta)$ a priori.

\section{Details on potential quantum advantage}\label{sec:Proof_Error_Model}

Because VQCs are parametrized in a different way from  classical models, they do not necessarily converge towards the MNLS as a result of the training. The weight vector associated to the quantum model may indeed not be contained in the row space of the data, as it is the case for classical linear regression.
Contrary to the classical case, one does not have access directly to the coefficients $\beta$ while tuning a quantum model. One instead optimizes a vector of parameters $\theta$ such that $\beta = \beta(\theta)$ and optimizes the loss function $\mathcal{L}(\theta) = \norm{y - X\beta(\theta)}^2$.
The update rule defined in \autoref{eqn:gd} becomes
\begin{align}
\label{eqn:gdq}
    \beta_{k+1} &= \beta(\theta_{k+1}) = \beta\Big(\theta_k - t\frac{\partial \mathcal{L}}{\partial \theta}\Big|_{\theta=\theta_k}\Big)\\
    \frac{\partial \mathcal{L}}{\partial \theta} &= \frac{\partial \beta}{\partial\theta}^\top X^\top (X\beta - y)
\end{align}
If $t$ is small enough, one can linearize \autoref{eqn:gdq} and write
\begin{align}
    \beta_{k+1} &= \beta_k - t\frac{\partial \beta}{\partial \theta}\Big|_{\theta=\theta_k}^\top \frac{\partial \mathcal{L}}{\partial \theta}\Big|_{\theta=\theta_k}\\
    &= \beta_k - t\frac{\partial \beta}{\partial \theta}\Big|_{\theta=\theta_k}^\top \frac{\partial \beta}{\partial \theta}\Big|_{\theta=\theta_k}X^\top (X\beta - y)
\end{align}
In the general case, $\beta$ does not remain in the row space of $X$ therefore it does not necessarily converge towards the minimum least square estimator. This constitutes a crucial potential distinction between quantum and classical models. It remains to see when $\beta_Q$ can converge far from $\beta_{\text{MNLS}}$ or an approximation of MNLS via Random Features.

We will also explain the way in which quantifying the weight vector norm of the quantum models ($\norm{\beta_Q}_2$) can serve as proxy to quantify the difference between a quantum model $f_{\beta_Q}$ and the classical solution $f_{\beta_{\rm MNLS}}$ (i.e. $\norm{f_{\beta_Q} - f_{\beta_{\rm MNLS}}}_{\infty}$), and hence the possibility of a potential quantum advantage; see \cref{thm:looking-at-norm-betaq-to-assess-potential-quantum-advantage}.

We adopt here a slightly more general setup than the one of the main text. Namely, we consider the following:
	\begin{itemize}
		
		\item $(\mathcal{V},\langle  \cdot , \cdot \cdot \rangle)$ is a finite-dimensional Hilbert space over $\mathbb{K}$  ($\mathbb{K}=\mathbb{R}$ or $\mathbb{C}$) --- the \emph{feature space} --- of $\dim_{\mathbb{K}}(\mathcal{V}):=p$. (The main text essentially takes the case $\mathcal{V}=\mathbb{R}^p$, along with its standard inner-product $\langle v, w \rangle = v^T w$.)
		\item $\mathcal{X}$ is any set. (The main text takes the case of a subset $\mathcal{X}$ of $\mathbb{R}^d$.)
		\item $\phi:\mathcal{X} \to \mathcal{V}$ is any bounded function, called the \emph{feature map}.
	\end{itemize}
	Furthermore, we denote:
	\begin{itemize}
		\item $L^{\infty}(\mathcal{X})$ the space of bounded functions from $\mathcal{X}$ to $\mathbb{K}$,
		\item for $\beta\in \mathcal{V}$: the map $f_{\beta}:\mathcal{X}\to \mathbb{K}$ defined by $f_\beta(x):= \langle \beta , \phi(x) \rangle$,
		\item $\mathtt{f}:\mathcal{V} \to L^{\infty}(\mathcal{X})$ the map that sends a weight vector to its associated linear model in feature space, i.e. $\mathtt{f}(\beta):=f_{\beta}$,
		\item for $k\geq1,x_{1},\dots,x_{k} \in \mathcal{X}$: the $k \times k$ $\mathbb{K}$-matrix $K_{x_{1},\dots,x_{k}}$ defined by: $[K_{x_{1},\dots,x_{k}}]_{ij}:= \langle \phi(x_{i}) , \phi(x_{j}) \rangle$,
		\item for $\beta\in \mathcal{V}$: $\norm{\beta}_2 := \sqrt{\langle \beta,\beta \rangle}$.
	\end{itemize}

	\medskip
	
	\begin{lemma}[Lipschitz continuity of $\mathtt{f}$ -- w.r.t. $\lVert \cdot \rVert_{2},\,\lVert \cdot \rVert_{\infty}$]\label{lem:lipschitz}
		For all $\beta,\beta'\in \mathcal{V}$, we have:
		\begin{equation}
			\lVert f_\beta - f_{\beta'}\rVert_{\infty} \leq \sup_{x \in \mathcal{X}}\!\big(\lVert \phi(x)\rVert_{2}\big) \ \lVert \beta - \beta'\rVert_{2}\,.
		\end{equation}
		\begin{proof}
			It follows from the fact that for any $x\in\mathcal{X}$,
			\begin{align}
				\lvert f_{\beta}(x) - f_{\beta'}(x)\rvert
				=  \big\lvert \langle  \beta - \beta' , \phi(x) \rangle \big\rvert
				\leq  \lVert \beta - \beta'\rVert_{2} \  \lVert \phi(x)\rVert_{2},
			\end{align}
			where the Cauchy-Schwartz inequality was used.
		\end{proof}
	\end{lemma}
	
	\begin{lemma}[Reverse Lipschitz continuity of $\mathtt{f}$ -- w.r.t. $\lVert \cdot \rVert_{2},\ \lVert \cdot \rVert_{\infty}$]\label{lem:reverse-lipschitz}
		Suppose that $\operatorname{span}(\phi(\mathcal{X}))=\mathcal{V}$. Then for any $x_{1},\dots,x_{p}\in \mathcal{X}$ such that $\phi(x_1),\dots,\phi(x_p)$ are linearly independent, we have, for all $\beta,\beta'\in \mathcal{V}$:
		
		\begin{equation}
			\lVert \beta - \beta'\rVert_{2} \leq \sqrt{p} \ \,\lambda_{\mathrm{min}}\!\left( K_{x_{1},\dots,x_{p}}\right)^{-1/2} \ \,\lVert f_\beta - f_{\beta'}\rVert_{\infty}\,,
		\end{equation}
		where $\lambda_{\mathrm{min}}(\cdot)$ denotes the smallest eigenvalue.
		
		\begin{proof} Let $x_{1},\dots,x_{p}\in \mathcal{X}$ be such that $\phi(x_1),\dots,\phi(x_p)$ are linearly independent. Let $\mathcal{B}$ be an orthonormal basis of $\mathcal{V}$, and denote by $X$ the $p \times p$ $\mathbb{K}$-matrix giving in columns the vectors $\phi(x_i)$ expressed in the basis $\mathcal{B}$, i.e.:
			\begin{align}\label{eq:matrix-of-Phi-xis}
				X :=
				\begin{bmatrix}
					| & | & & | \\
					\big(\phi(x_1)\big)_{\mathcal{B}} & \big(\phi(x_2)\big)_{\mathcal{B}} & \cdots & \big(\phi(x_p)\big)_{\mathcal{B}} \\
					| & | & & |
				\end{bmatrix}.
			\end{align}
			Note that this matrix $X$ is invertible by assumption. For any $\beta \in \mathcal{V}$, denote also by $\vec{\beta} \in \mathbb{K}^p$ the column vector of its components in the basis $\mathcal{B}$.
			
			For any $\beta,\beta' \in \mathcal{V}$, we have:
			\begin{align}
				\lVert \beta - \beta' \rVert_{2}
				&= \lVert \vec{\beta} - \vec{\beta'} \rVert_{2}\\
				&= \lVert (X^{\dagger})^{-1} \,  X^{\dagger} \, (\vec{\beta} - \vec{\beta'}) \rVert_{2} \label{eq:reverse-lipschitz-proof--eq1}\\
				&\leq \lVert (X^{\dagger})^{-1}\rVert_{\infty} \   \quad\lVert X^{\dagger} \, (\vec{\beta} - \vec{\beta'}) \rVert_{2} \label{eq:reverse-lipschitz-proof--eq2}\\
				&\leq \lVert (X^{\dagger})^{-1}\rVert_{\infty} \   \quad\sqrt{p} \ \lVert X^{\dagger} \, (\vec{\beta} - \vec{\beta'}) \rVert_{\infty} \label{eq:reverse-lipschitz-proof--eq3}\\
				&\leq \lVert (X^{\dagger})^{-1}\rVert_{\infty} \   \quad\sqrt{p} \ \lVert f_{\beta} - f_{\beta'}\rVert_{\infty} \label{eq:reverse-lipschitz-proof--eq4}\\
				&= \lambda_{\mathrm{min}}\!\left( K_{x_{1},\dots,x_{k}}\right)^{-1/2} \,  \,\sqrt{p} \ \lVert f_{\beta} - f_{\beta'}\rVert_{\infty}. \label{eq:reverse-lipschitz-proof--eq5}
			\end{align}
			In the above, the equality in \labelcref{eq:reverse-lipschitz-proof--eq1} is valid since $X^\dagger$ is (left-)invertible (since $X$ is invertible). The inequality in \labelcref{eq:reverse-lipschitz-proof--eq2} follows from the "compatibility" inequality $\lVert A v \rVert_{2} \leq \lVert A \rVert_{\infty} \lVert v \rVert_{2}$ (which holds by the fact that the spectral norm $\lVert \cdot \rVert_{\infty}$ is the operator norm induced by the vector $2$-norm). The inequality in \labelcref{eq:reverse-lipschitz-proof--eq3} uses the relation $\lVert v \rVert_{2} \leq \sqrt{p} \lVert v \rVert_{\infty}$ between the vector $2$-norm and the vector $\infty$-norm. The inequality in \labelcref{eq:reverse-lipschitz-proof--eq4} follows from the definitions, since $\lVert X^{\dagger} \, (\vec{\beta} - \vec{\beta'}) \rVert_{\infty} := \max_{i=1,\dots,p} \big\lvert \langle  \beta - \beta' , \phi(x_i) \rangle \big\rvert \leq \sup_{x \in \mathcal{X}} \big\lvert \langle  \beta - \beta' , \phi(x) \rangle \big\rvert$. Lastly, the equality in \labelcref{eq:reverse-lipschitz-proof--eq5} follows from the general fact that for any invertible matrix $A$, 
			$\lVert (A^{\dagger})^{-1}\rVert_{\infty} = \lVert ((A^{\dagger})^{-1})^\dagger ((A^{\dagger})^{-1}) \rVert_{\infty}^{1/2} = \lVert A^{-1} ((A^{\dagger})^{-1})\rVert_{\infty}^{1/2} =  \lVert (A^\dagger A)^{-1} \rVert_{\infty}^{1/2} = 
			\lambda_{\mathrm{min}}(A^\dagger A)^{-1/2}$, and the fact that $X^\dagger X = K_{x_{1},\dots,x_{p}}$.
		\end{proof}
	\end{lemma}
	
	\begin{lemma}\label{lem:entry-wise-linindep-featuremap}
		For the case of the feature space $\mathcal{V}:=\mathbb{K}^p$
		and a feature map $\phi:\mathcal{X}\to\mathbb{K}^p$ of the form
		\begin{align}
			\phi(x) &=
			\begin{bmatrix}
				g_{1}(x) \\
				g_{2}(x) \\
				\vdots \\
				g_{p}(x)
			\end{bmatrix}, 
		\end{align}
		where $\mathcal{X}$ is a vector space, and $g_1,\dots,g_p: \mathcal{X} \to \mathbb{K}$ are \emph{linearly independent functions}, it holds that:
		\begin{align}
			\operatorname{span}(\phi(\mathcal{X}))=\mathcal{V}.
		\end{align}
		\begin{proof}
			We proceed by induction on $p \in \mathbb{N}_{\geq 1}$.
			
			For $p=1$, the claim is that there exists $x \in \mathcal{X}$ such that $g_1(x) \neq 0$, but this is true as the hypothesis of "$(g_1)$ being linearly independent" says that $g_1$ is not identically zero.
			
			Suppose the claim holds for $p\geq1$, i.e., there exists $x_1,\dots,x_p \in \mathcal{X}$ such that $\phi(x_1),\dots,\phi(x_p)$ are linearly independent, or in other words, such that the matrix 
			\begin{align}
				X(x_1,\dots,x_p) := 
				\begin{bmatrix}
					g_1(x_1) &  \cdots & g_1(x_p) \\
					\vdots & \ddots & \vdots \\
					g_p(x_1) & \cdots & g_p(x_p)
				\end{bmatrix}
			\end{align}
			has nonzero determinant. Let us show that there exists a value of $x_{p+1} \in \mathcal{X}$ such that the matrix
			\begin{align}
				X(x_1,\dots,x_p, x_{p+1} ) := 
				\begin{bmatrix}
					g_1(x_1) &  \cdots & g_1(x_p) & g_1(x_{p+1}) \\
					\vdots & \ddots & \vdots & \vdots \\
					g_p(x_1) & \cdots & g_p(x_p) & g_p(x_{p+1}) \\
					g_{p+1}(x_1) & \cdots & g_{p+1}(x_p) & g_{p+1}(x_{p+1})
				\end{bmatrix}
			\end{align}
			has nonzero determinant as well.
			Writing the cofactor expansion along the last column of the above matrix's determinant provides:
			\begin{align}\label{eq:entry-wise-linindep-featuremap-proof--expansion}
				\det\!\big[ X(x_1,\dots,x_p, x_{p+1}) \big]
				= \alpha_{1} g_1(x_{p+1}) + \dots +  \alpha_{p} g_p(x_{p+1}) + \alpha_{p+1} g_{p+1}(x_{p+1}),
			\end{align}
			where $\alpha_1,\dots,\alpha_p,\alpha_{p+1}$ are scalars that depend on $x_1,\dots,x_p$ but not on $x_{p+1}$, and we have in particular $\alpha_{p+1} = (-1)^{p+1} \det\!\big[ X(x_1,\dots,x_p) \big]$.
			
			Suppose by contradiction that no $x_{p+1}$ exists such that \cref{eq:entry-wise-linindep-featuremap-proof--expansion} is nonzero, i.e.:
			\begin{align}
				&\forall x_{p+1} \in \mathcal{X}\qquad 
				\alpha_{1} g_1(x_{p+1}) + \dots +  \alpha_{p} g_p(x_{p+1}) + \alpha_{p+1} g_{p+1}(x_{p+1})
				\ \ =\ \ 0
			\end{align}
			This means that we have the following null linear combination of functions on $\mathcal{X}$:
			\begin{equation}
				\alpha_{1} \, g_1 + \dots +  \alpha_{p} \, g_p + \alpha_{p+1} \, g_{p+1} \ = \ 0,
			\end{equation}
			and since $g_1,\dots,g_{p+1}$  are linearly independent functions, this implies that all these coefficients $\alpha_i$ are zero. In particular, the last one is zero, so
			\begin{equation}
				\det\!\big[ X(x_1,\dots,x_p) \big] = 0.
			\end{equation}
			But this is in contradiction with the starting hypothesis. Therefore, there must exits a $x_{p+1} \in \mathcal{X}$ such that \cref{eq:entry-wise-linindep-featuremap-proof--expansion} is nonzero, which concludes the proof.
		\end{proof}
	\end{lemma}
	\begin{remark}
		In the situation of \cref{lem:entry-wise-linindep-featuremap}, if $\mathcal{X}=\mathbb{R}^n$ and if the functions $g_1,\dots,g_p$ are (real or complex) \emph{analytic}, then it furthermore holds that the set
		\begin{equation}
			\big\{ (x_1,\dots,x_p) \in (\mathbb{R}^{k})^{p} \ \,|\,\ \phi(x_1),\dots,\phi(x_p)\text{ are linearly independent} \big\}
		\end{equation}
		has a complement of zero Lebesgue measure in $(\mathbb{R}^{k})^{p}$.
		
		In other words, we have in this case that \emph{almost all} $(x_1,\dots,x_p) \in \mathcal{X}^p$ satisfy the property that $\phi(x_1),\dots,\phi(x_p)$ are linearly independent.
		
		This is because in this case, the function $(x_1,\dots,x_p) \mapsto \lvert\det X\rvert^2$, where $X$ is the matrix of the vectors $\phi(x_1),\dots,\phi(x_p)$ (see \cref{eq:matrix-of-Phi-xis}), is real-analytic, and by assumption this function is non-zero on \emph{at least one} input $(x_1,\dots,x_p)$; so the remark follows by the fact that the zero set of an analytic and non identically zero function has measure zero \cite{mityagin2015zero}.
	\end{remark}
	
	\begin{lemma}[Fourier functions are linearly independent]\label{lem:Fourier-functs-are-lin-indep}
		For $p \geq 1$, and for $\mathcal{X}$ any subset of $\mathbb{R}$, it holds that:
		\begin{enumerate}
			\item  if $\omega_1,\dots,\omega_p \in \mathbb{R}$ are all distinct, then the functions $g_1,\dots,g_p: \mathcal{X} \to \mathbb{C}$ given by
			\begin{align}
				g_j(x) := e^{i \omega_{j} x}
			\end{align}
			are linearly independent functions; and
			
			\item if $\omega_1,\dots,\omega_p \in \mathbb{R}$ are all distinct and no two of them are negations of each other, then the functions $g_0, g_1,\dots,g_{2p}: \mathcal{X} \to \mathbb{R}$ given by
			\begin{align}
				g_{0}(x)    &:= 1,\\
				g_{2j}(x)   &:= \cos(\omega_{j} x),\\
				g_{2j+1}(x) &:= \sin(\omega_{j} x),
			\end{align}
			are linearly independent functions.
		\end{enumerate}
		\begin{proof}
			Denote by $C^{\infty}(\mathcal{X},\mathbb{K})$ the set of smooth functions from $\mathcal{X}$ to $\mathbb{K}$, and by $\mathcal{D}:C^{\infty}(\mathcal{X},\mathbb{K}) \to C^{\infty}(\mathcal{X},\mathbb{K})$ the derivative map (which is linear).
			
			For the first claim (here $\mathbb{K}=\mathbb{C}$), the functions $g_j$ are eigenvectors of $\mathcal{D}$ associated to distinct eigenvalues ($i \omega_j$), therefore they are mutually linearly independent.
			
			For the second claim (here $\mathbb{K}=\mathbb{R}$), the functions $g_{2j}, g_{2j+1}$ are eigenvectors of $\mathcal{D}^2$ associated to eigenvalues $(-\omega_j^2)$, which are distinct for different $j$'s since by assumption $\omega_j>0$ for all $j$, and furthermore $g_{2j}$ and $g_{2j+1}$ are linearly independent as well, and lastly $g_0$ is linearly independent from all the others as well; therefore one obtains that $g_0,\dots,g_{2q}$ are mutually linearly independent.
		\end{proof}
	\end{lemma}
	
	We state a generalization of this lemma to the multi-dimensional case, with the added constraint that the frequencies are integers ($\omega_j \in \mathbb{Z}^d$). Indeed, in this case, this lemma is then immediate by appealing to the orthogonality of the mentioned functions under the standard $L^2$ inner-product.

	\begin{lemma}[Multi-dimensional generalization of \cref{lem:Fourier-functs-are-lin-indep}]\label{lem:Multidim-Fourier-functs-are-lin-indep}
	For $d\geq1, p \geq 1$, and for $\mathcal{X}$ any subset of $\mathbb{Z}^d$, it holds that:
	\begin{enumerate}
		\item  if $\omega_1,\dots,\omega_p \in \mathbb{Z}^d$ are all distinct, then the functions $g_1,\dots,g_p: \mathcal{X} \to \mathbb{C}$ given by
		\begin{align}
			g_j(x) := e^{i \omega_{j}^T x}
		\end{align}
		are linearly independent functions; and
		
		\item if $\omega_1,\dots,\omega_p \in \mathbb{Z}^d$ are all distinct and no two of them are negations of each other, then the functions $g_0, g_1,\dots,g_{2p}: \mathcal{X} \to \mathbb{R}$ given by
		\begin{align}
			g_{0}(x)    &:= 1,\\
			g_{2j}(x)   &:= \cos(\omega_{j}^T x),\\
			g_{2j+1}(x) &:= \sin(\omega_{j}^T x),
		\end{align}
		are linearly independent functions.
	\end{enumerate}
\end{lemma}

Note that the assumption of integer frequencies could be relaxed to merely \emph{ $\omega_1,\dots,\omega_p \in \mathbb{R}^d$ are all distinct and no two of them are negations of each other} (a proof would require more care).

\begin{restatable}[Relation between the $\infty$-distance of two models and the $2$-norms of their weight vectors]{thm}{RelationBetaNormFunctionDist}\label{thm:RelationBetaNormFunctionDist}
Consider two arbitrary weight vectors, denoted suggestively as $\beta_C,\beta_Q\in \mathcal{V}$, and let $f_{\beta_C}$, $f_{\beta_Q}$ be the corresponding functions ($f_{\beta_i}(x)= \langle \beta_i , \phi(x) \rangle$). Suppose that the feature map $\phi$ satisfies $\operatorname{span}(\phi(\mathcal{X}))=\mathcal{V}$. Then, the following holds:

\begin{align}\label{eq:RelationBetaNormFunctionDist-eq}
    \norm{f_{\beta_Q} - f_{\beta_C}}_{\infty} \ &\leq\ L\norm{\beta_Q - \beta_C}_2\,,\\[8pt]
    \label{eq:RelationBetaNormFunctionDist-eq-2}
    \Big| \norm{\beta_Q}_2 - \norm{\beta_C}_2 \Big|\ &\leq\  c\, \norm{f_{\beta_Q} - f_{\beta_C}}_{\infty}\,,
\end{align}
where $L:=\sup_{x \in \mathcal{X}}\!\big(\lVert \phi(x)\rVert_{2}\big)$, $c := \sqrt{p} \ \,\lambda_{\mathrm{min}}\!\left( K_{x_{1},\dots,x_{p}}\right)^{-1/2}$, with $x_{1},\dots,x_{p}\in \mathcal{X}$ being any set of $p$ points of $\mathcal{X}$ such that $\phi(x_1),\dots,\phi(x_p)$ are linearly independent, $K_{x_{1},\dots,x_{p}}$ is the the kernel matrix defined by $[K_{x_{1},\dots,x_{p}}]_{ij}:= \phi(x_{i})^T \phi(x_{j})$, and where $\lambda_{\mathrm{min}}(A)$ denotes the smallest eigenvalue of a matrix $A$.

In particular, if $f_{\beta_Q}$ corresponds to a quantum model that has been (perfectly) trained on a dataset $(x_1,\dots,x_M)$ of size $M<p$ to a target function $f^*$ --- i.e. it is a solution that achieves zero empirical risk $\mathcal{L}(f_{\beta_Q},f^*):= \frac{1}{M}\sum_{i=1}^M (f_{\beta_Q}(x_i) - f^*(x_i))^2$ --- and if $f_{\beta_C}$ is the corresponding classical solution $f_{\beta_{\mathrm{MNLS}}}$ obtained through classical linear regression, then the above hypothesis apply and \cref{eq:RelationBetaNormFunctionDist-eq,eq:RelationBetaNormFunctionDist-eq-2} reduce to:
\begin{align}\label{eq:RelationBetaNormFunctionDist-eq-case-of-solutions}
    \norm{f_{\beta_Q} - f_{\beta_{\rm MNLS}}}_{\infty} \ &\leq\ \sqrt{\norm{\beta_Q}_2^2 - \norm{\beta_{\rm MNLS}}_2^2 }\,,\\[8pt]
    \label{eq:RelationBetaNormFunctionDist-eq-2-case-of-solutions}
    \norm{\beta_Q}_2 - \norm{\beta_{\rm MNLS}}_2\ &\leq\  c\, \norm{f_{\beta_Q} - f_{\beta_{\rm MNLS}}}_{\infty}\,.
\end{align}
\end{restatable}

\begin{proof}
Let $\beta_C,\beta_Q \in \mathcal{V}$. \Cref{eq:RelationBetaNormFunctionDist-eq} holds by \cref{lem:lipschitz}, while \cref{eq:RelationBetaNormFunctionDist-eq-2} holds because
\begin{equation}
\Big| \norm{\beta_Q}_2 - \norm{\beta_C}_2 \Big| \leq \lVert \beta_Q - \beta_C\rVert_{2} \leq c \ \lVert f_{\beta_Q}- f_{\beta_C}\rVert_{\infty}\,,
\end{equation}
where the first inequality is the reverse triangle inequality, and the second holds by \cref{lem:reverse-lipschitz}.

Now, assume additionally that $f_{\beta_Q}$ and $f_{\beta_C}:=f_{\beta_{\rm MNLS}}$ correspond to (perfectly) trained quantum and respective classical models, in the overparameterized regime ($p>M$).

First, note that here the feature map $\phi$ (\cref{eq:feat-map}) satisfies the hypothesis $\operatorname{span}(\phi(\mathcal{X}))=\mathcal{V}$, because they are of the form of \cref{lem:entry-wise-linindep-featuremap} with their entries $g_k$ being of the form of \cref{lem:Multidim-Fourier-functs-are-lin-indep}; and furthermore it satisfies
$L = \sup_{x \in \mathcal{X}}(\lVert \phi(x)\rVert_{2}) = 1$.

Second, we recall the following facts about the space of weight vectors $\beta$ that are solutions to $\mathcal{L}(f_{\beta},f^*)=0$ (see also \cref{app:mnls} above): it is an affine subspace (of dimension $p-M$) inside the feature space $\mathcal{V}$, it contains the vector $\beta_{\rm MNLS}$ (its minimum norm vector), and it is orthogonal to that vector. 

Since $\beta_Q$ is such a solution weight vector, we therefore have
\begin{equation}\label{eq:beta-q-orhtog-decomp-proprety}
\beta_Q = \beta_{\rm MNLS} + u\,,\qquad \text{for some $u \in \mathbb{R}^p$ with } u^\top  \beta_{\rm MNLS} = 0\,.
\end{equation}
From \labelcref{eq:beta-q-orhtog-decomp-proprety} it follows (Pythagorean theorem) that $\norm{\beta_Q - \beta_{\rm MNLS}}_2 = \sqrt{\norm{\beta_Q}_2^2 - \norm{\beta_{\rm MNLS}}_2^2}$, which justifies why \cref{eq:RelationBetaNormFunctionDist-eq} can be rewritten as \cref{eq:RelationBetaNormFunctionDist-eq-case-of-solutions}. It also follows from \labelcref{eq:beta-q-orhtog-decomp-proprety} that $\norm{\beta_Q}_2 \geq \norm{\beta_{\rm MNLS}}_2$, justifying why \cref{eq:RelationBetaNormFunctionDist-eq-2} can be rewritten as \cref{eq:RelationBetaNormFunctionDist-eq-2-case-of-solutions}.
\end{proof}

The above \cref{thm:RelationBetaNormFunctionDist} can therefore be interpreted in the following way:

\begin{restatable}[Large weight vector norm and quantum advantage]{cor}{Looking-at-norm-betaq-to-assess-potential-quantum-advantage}\label{thm:looking-at-norm-betaq-to-assess-potential-quantum-advantage}
Suppose that the quantum model $f_{\beta_Q}$ has been (perfectly) trained on a dataset $(x_1,\dots,x_M)$ of size $M>p$ to a target function $f^*$ --- i.e. it is a solution that achieves zero empirical risk $\mathcal{L}(f_{\beta_Q},f^*):= \frac{1}{M}\sum_{i=1}^M (f_{\beta_Q}(x_i) - f^*(x_i))^2$ ---, and let $f_{\beta_{\mathrm{MNLS}}}$ be the solution obtained through classical linear regression.

Let us refer to having a large value of $\norm{f_{\beta_{\mathrm{MNLS}}} - f_{\beta_Q}}_{\infty}$ as having a ``potential quantum advantage'', and a small value as an ``absence of quantum advantage''.

Then, the two inequalities of \cref{eq:RelationBetaNormFunctionDist-eq-case-of-solutions,eq:RelationBetaNormFunctionDist-eq-2-case-of-solutions} may be respectively read as follows:
\begin{enumerate}
    \item Having a small weight vector norm $\norm{\beta_Q}_2$ (close to the value of $\norm{\beta_{\rm MNLS}}_2$) implies an absence of quantum advantage. 

    \item Provided that $c$ is small, having a large weight vector norm $\norm{\beta_Q}_2$ (compared to the value of $\norm{\beta_{\mathrm{MNLS}}}_2$) implies a potential quantum advantage.
\end{enumerate}
\end{restatable}

\section{Concentration of eigenvalues of the kernel matrix for the Fourier feature map with integer coefficients}\label{app:Concentration_Eigen_K_Fourier_integer}

In this section, we establish concentration bounds on the eigenvalues of the kernel matrix. We consider the Fourier feature map:
\begin{equation}
    \phi(x) = \frac{1}{\sqrt{p}}\begin{bmatrix} \cos(\omega^{\top}x) \\ \sin(\omega^{\top}x) \end{bmatrix}
\end{equation}

with $\Omega \subset \mathbb{Z}^d$. 

The kernel is defined as:
\begin{equation}
    \begin{split}
        k(x,x') &= \phi(x)^\top \phi(x')\\
        &= \frac{1}{p} \sum_{\omega \in \Omega} \cos(\omega^{\top}x) \cos(\omega^{\top}x')  + \sin(\omega^{\top}x) \sin(\omega^{\top}x') \\
    \end{split}
\end{equation}

\begin{restatable}[]{thm}{AnalysisKernelMatrix}
\label{thm:AnalysisKernelMatrix}
    Let $(x_1, \dots x_M)$ uniformly distributed on $[0, 2\pi]^d$, and let $K$ be the empirical kernel matrix. Then there exists a constant $C$ such that 
    \begin{equation}
        \PP(\lambda_{\min}(K)> \frac{1}{2}) \geq 1- \frac{C}{p} \frac{(M-1)^2p^2}{(p-4(M-1)^2)^2} = 1 -\frac{C}{p}\frac{(M-1)^2}{1 - 4\frac{(M-1)^2}{p^2}}
    \end{equation}
\end{restatable}

\begin{proof}
The proof is an application of \cref{lemma:eigenval_trace}, \cref{lemma:exp_var_s} and 
Chebyshev's inequality. 
Let $s^2 = 1 + \frac{1}{M} \sum_{j \neq i}^M k(x_i,x_i)^2 - 1 = \frac{1}{M} \sum_{j \neq i}^M k(x_i,x_j)^2$.
    \begin{align}
    \PP(|s^2 - \frac{2(M-1)}{p}| > \epsilon) &\leq \frac{\VV[s^2]}{\epsilon^2}\\
    \PP(s^2 \geq \frac{1}{4(M-1)}) &\leq \PP(s^2- \frac{2(M-1)}{p}\geq \frac{1}{4(M-1)}- \frac{2(M-1)}{p}) \\
    &\leq \PP(s^2- \frac{(p-8(M-1)^2)}{4(M-1)p})\\
    & \leq \frac{C}{p}\frac{4(M-1)^2p^2}{(p-8(M-1)^2)^2}
    \end{align}
    With lemma \ref{lemma:eigenval_trace}, 
    \begin{align}
        \PP(\lambda_{\min}(K)> \frac{1}{2}) &\geq \PP(s \leq \frac{1}{2\sqrt{M-1}}) =  \PP(s^2 \leq \frac{1}{4(M-1)})\\
        &\geq 1 - \PP(s^2 \geq \frac{1}{4(M-1)})\\
        &\geq 1 - \frac{C}{p}\frac{4(M-1)^2p^2}{(p-8(M-1)^2)^2}
    \end{align}

\end{proof}

Let $(x_1, \dots x_M)$ uniformly distributed on $[0, 2\pi]^d$, and let $K$ be the empirical kernel matrix. We want to lower bound the smallest eigenvalue of K, and we use the following lemma 
\begin{restatable}[]{lemma}{}
\label{lemma:eigenval_trace}
    Theorem 2.1 from \cite{WOLKOWICZ1980471}.
    Let $A$ be a $M\times M$ complex matrix with real eigenvalues. Let $m = \text{tr}(A)/M$  and  $s^2 = \text{tr}(A^2)/M - m^2$. We have that 
    \begin{equation}
        m - s\sqrt{M-1} \leq \lambda_{\text{min}}(A) \leq m - \frac{s}{\sqrt{M-1}}
    \end{equation}
\end{restatable}

For the kernel matrix, $m = 1$ and $s^2 = 1 + \frac{1}{M} \sum_{j \neq i}^M k(x_i,x_i)^2 - 1 = \frac{1}{M} \sum_{j \neq i}^M k(x_i,x_j)^2$.

We prove the following result:
\begin{restatable}[]{lemma}{}
\label{lemma:exp_var_s}
The expectation of $s^2$ is given by
    \begin{itemize}
        \item $\EE[s^2] = \displaystyle\frac{M(M-1)}{2Mp} = \displaystyle\frac{(M-1)}{2p}$
    \end{itemize}

Furthermore, there exists a constant $C$ such that 
\begin{itemize}
        \item $\VV[s^2] \leq \displaystyle\frac{C}{p}$
    \end{itemize}
\end{restatable}

\begin{proof}
    With the use of \autoref{lemma:kernel_elements}, we have that
\begin{align}
    \EE[s^2] = \frac{M(M-1)}{M}\EE(k(x, x')^2) = \frac{M-1}{2p}
\end{align}
\begin{align}
\EE[s^4] &= \EE\Bigg[\frac{1}{M^2} \bigg(\sum_{j \neq i}^M k(x_i,x_j)^2\bigg)^2\Bigg] = \frac{1}{M^2}\EE\Bigg[ 2\sum_{j \neq i}^M k(x_i,x_j)^4 + \sum_{j \neq i\neq k\neq l=1}^M k(x_i,x_j)^2k(x_k,x_l)^2\Bigg]\\
&\leq \frac{C'}{p}\frac{M(M-1)}{M^2} + \frac{M(M-1)(M-2)(M-3)}{M^2}\frac{(M-1)^2}{4p^2}\\
\VV[s^2] &\leq \frac{C'}{p}\frac{M(M-1)}{M^2} + \frac{M(M-1)(M-2)(M-3)}{M^2}\frac{1}{4p^2} - \frac{(M-1)^2}{4p^2}\\
&\leq \frac{C'}{p}\frac{M(M-1)}{M^2} +\frac{1}{4p^2}((M-2)(M-3) - (M-1)^2)\\
&\leq \frac{C}{p}
\end{align}
\end{proof}

We just need to prove the following lemma
\begin{restatable}[]{lemma}{}
 \label{lemma:kernel_elements}
Let $x_i$, $x_j$ independent and uniformly distributed on $[0, 2\pi]^d$. There exists a constant $C'>0$ such that 
\begin{itemize}
        \item $\mathbb{E}[k(x_i,x_j)^2] = \displaystyle\frac{1}{2p}$
        \item $\mathbb{E}[k(x_i,x_j)^4] \leq \displaystyle\frac{C'}{p}$
    \end{itemize}
\end{restatable}

\begin{align}
        \mathbb{E}[k(x_i,x_j)^2] &= \frac{1}{p^2} \mathbb{E}[ \Sum_{\omega}(\cos(\omega^{\top}x_i) \cos(\omega^{\top}x_j) + \sin(\omega^{\top}x_i) \sin(\omega^{\top}x_j))^2] \\
        &+ \frac{1}{p^2} \mathbb{E}[\Sum_{\omega \neq \omega'}(\cos(\omega^{\top}x_i) \cos(\omega^{\top}x_j)
        + \sin(\omega^{\top}x_i) \sin(\omega^{\top}x_j)) \times\\
        &(\cos(\omega'^{\top}x_i) \cos(\omega'^{\top}x_j)
        + \sin(\omega'^{\top}x_i) \sin(\omega'^{\top}x_j))]  
\end{align}

We have that $ \mathbb{E}[\cos^2(\omega^{\top}x_i)\cos^2(\omega^{\top}x_j)] = \frac{1}{4}$
 and $\mathbb{E}[\sin^2(\omega^{\top}x_i)\sin^2(\omega^{\top}x_j)] = \frac{1}{4}$.
 
 Furthermore $\mathbb{E}[\cos(\omega^{\top}x_i) \cos(\omega^{\top}x_j) \sin(\omega^{\top}x_i) \sin(\omega^{\top}x_j)] = 0$.

Therefore,
\begin{equation}
    \mathbb{E}[k(x_i,x_j)^2] = \frac{1}{2p}
\end{equation}

We introduce some simplifying notations that will be used for the rest of the calculation.

\begin{align}
    &\sum_{\omega \in \Omega} c_ic_j + s_is_j = \sum_{\omega \in \Omega} \cos(\omega^{\top}x_i) \cos(\omega^{\top}x_j)  + \sin(\omega^{\top}x_i) \sin(\omega^{\top}x_j)\\
    &\sum_{\omega_1, \omega_2 \in \Omega} (c_i^{(1)}c_j^{(1)} + s_i^{(1)}s_j^{(1)})(c_i^{(2)}c_j^{(2)} + s_i^{(2)}s_j^{(2)}) \\ &= \sum_{\omega_1, \omega_2 \in \Omega}\Big(\cos(\omega_1^{\top}x_i) \cos(\omega_1^{\top}x_j)  + \sin(\omega_1^{\top}x_i) \sin(\omega_1^{\top}x_j))\Big)
    \Big(\cos(\omega_2^{\top}x_i) \cos(\omega_2^{\top}x_j)  + \sin(\omega_2^{\top}x_i) \sin(\omega_2^{\top}x_j)\Big)
\end{align}

By applying the multinomial formulas, we have that
\begin{align}
    &k(x_i,x_j)^4 = \frac{1}{p^4}\bigg(\sum_{\omega \in \Omega} c_ic_j + s_is_j\bigg)^4\\
    &=\frac{1}{p^4}\sum_{\omega \in \Omega} (c_ic_j + s_is_j)^4\\
    &+ \frac{4}{p^4}\sum_{\omega_1, \omega_2 \in \Omega, \omega_1 \neq \omega_2} (c_i^{(1)}c_j^{(1)} + s_i^{(1)}s_j^{(1)})^3(c_i^{(2)}c_j^{(2)} + s_i^{(2)}s_j^{(2)})\\
    &+ \frac{6}{p^4}\sum_{\omega_1, \omega_2 \in \Omega, \omega_1 \neq \omega_2} (c_i^{(1)}c_j^{(1)} + s_i^{(1)}s_j^{(1)})^2(c_i^{(2)}c_j^{(2)} + s_i^{(2)}s_j^{(2)})^2\\
    &+ \frac{4}{p^4}\sum_{\omega_1, \omega_2, \omega_3 \in \Omega, \omega_1 \neq \omega_2 \neq \omega_3} (c_i^{(1)}c_j^{(1)} + s_i^{(1)}s_j^{(1)})^2(c_i^{(2)}c_j^{(2)} + s_i^{(2)}s_j^{(2)})(c_i^{(3)}c_j^{(3)} + s_i^{(3)}s_j^{(3)})\\
    &+ \frac{24}{p^4}\sum_{\omega_1, \omega_2, \omega_3, \omega_4 \in \Omega, \omega_1 \neq \omega_2 \neq \omega_3 \neq \omega_4} (c_i^{(1)}c_j^{(1)} + s_i^{(1)}s_j^{(1)})(c_i^{(2)}c_j^{(2)} + s_i^{(2)}s_j^{(2)})(c_i^{(3)}c_j^{(3)} + s_i^{(3)}s_j^{(3)})(c_i^{(4)}c_j^{(4)} + s_i^{(4)}s_j^{(4)})
\end{align}

Our goal is to show that there exists a constant $C$ such that $\EE[k(x_i, x_j)^4] \leq \frac{C}{p}$. All the sums except the last one have at most $p^3$ terms, and we can bound each one of them using the fact that $|c_ic_j + s_is_j| \leq 2$. We need to look at the last sum.

\begin{align}
    &\EE[(c_i^{(1)}c_j^{(1)} + s_i^{(1)}s_j^{(1)})(c_i^{(2)}c_j^{(2)} + s_i^{(2)}s_j^{(2)})(c_i^{(3)}c_j^{(3)} + s_i^{(3)}s_j^{(3)})(c_i^{(4)}c_j^{(4)} + s_i^{(4)}s_j^{(4)})]\\
    &= \sum_{f_1, f_2, f_3, f_4 \in \{\cos, \sin\}} \EE[f_1(\omega_1^\top x_i)f_2(\omega_2^\top x_i)f_3(\omega_3^\top x_i)f_4(\omega_4^\top x_i)f_1(\omega_1^\top x_j)f_2(\omega_2^\top x_j)f_3(\omega_3^\top x_j)f_4(\omega_4^\top x_j)]\\
    &= \sum_{f_1, f_2, f_3, f_4 \in \{\cos, \sin\}} \EE[f_1(\omega_1^\top x_i)f_2(\omega_2^\top x_i)f_3(\omega_3^\top x_i)f_4(\omega_4^\top x_i)]^2
\end{align}

\begin{align*}
    \cos(\omega_1^\top x_i)\cos(\omega_2^\top x_i)\cos(\omega_3^\top x_i)\cos(\omega_4^\top x_i)
    = \frac{1}{8}[&\cos((\omega_1 + \omega_2 + \omega_3 + \omega_4)^\top x_i)\\
    +&\cos((\omega_1 + \omega_2 - \omega_3 - \omega_4)^\top x_i)\\ 
    +&\cos((\omega_1 + \omega_2 + \omega_3 - \omega_4)^\top x_i)\\ 
    +&\cos((\omega_1 + \omega_2 - \omega_3 + \omega_4)^\top x_i)\\ 
    +&\cos((\omega_1 - \omega_2 + \omega_3 + \omega_4)^\top x_i)\\ 
    +&\cos((\omega_1 - \omega_2 - \omega_3 - \omega_4)^\top x_i)\\ 
    +&\cos((\omega_1 - \omega_2 + \omega_3 - \omega_4)^\top x_i)\\ 
    +&\cos((\omega_1 - \omega_2 - \omega_3 + \omega_4)^\top x_i)\\ 
    &]\\
    \cos(\omega_1^\top x_i)\cos(\omega_2^\top x_i)\cos(\omega_3^\top x_i)\sin(\omega_4^\top x_i)
    = \frac{1}{8}[&\sin((\omega_1 + \omega_2 + \omega_3 + \omega_4)^\top x_i)\\
    -&\sin((\omega_1 + \omega_2 - \omega_3 - \omega_4)^\top x_i)\\ 
    -&\sin((\omega_1 + \omega_2 + \omega_3 - \omega_4)^\top x_i)\\ 
    +&\sin((\omega_1 + \omega_2 - \omega_3 + \omega_4)^\top x_i)\\ 
    +&\sin((\omega_1 - \omega_2 + \omega_3 + \omega_4)^\top x_i)\\ 
    -&\sin((\omega_1 - \omega_2 - \omega_3 - \omega_4)^\top x_i)\\ 
    -&\sin((\omega_1 - \omega_2 + \omega_3 - \omega_4)^\top x_i)\\ 
    +&\sin((\omega_1 - \omega_2 - \omega_3 + \omega_4)^\top x_i)\\ 
    &]\\
\end{align*}
All terms in $\sin(\omega^\top x_i) \: $for all $ \omega$ and all terms in $\cos(\omega^\top x_i) $ for all $ \omega \neq 0$  have a zero expectation, therefore the only contribution comes from the terms such that there is an even number of $\sin$ among $\{f_1, f_2, f_3, f_4\}$.

Given $\omega_1, \omega_2, \omega_3$, there is only a constant number of $\omega_4$ such that there is a zero element in the set $\{\omega_1 + \epsilon\omega_2 + \epsilon'\omega_3 + \epsilon''\omega_4, \epsilon, \epsilon', \epsilon'' \in\{+1, -1\}\}$.

Therefore 

\begin{align}
    &\sum_{\omega_4 \in \Omega \backslash \{\omega_1, \omega_2, \omega_3\}}\sum_{f_1, f_2, f_3, f_4 \in \{\cos, \sin\}} \EE[f_1(\omega_1^\top x_i)f_2(\omega_2^\top x_i)f_3(\omega_3^\top x_i)f_4(\omega_4^\top x_i)]^2 \leq C\\
    &\frac{24}{|\Omega|^4}\sum_{\omega_1, \omega_2, \omega_3, \omega_4 \in \Omega, \omega_1 \neq \omega_2 \neq \omega_3 \neq \omega_4} (c_i^{(1)}c_j^{(1)} + s_i^{(1)}s_j^{(1)})(c_i^{(2)}c_j^{(2)} + s_i^{(2)}s_j^{(2)})(c_i^{(3)}c_j^{(3)} + s_i^{(3)}s_j^{(3)})(c_i^{(4)}c_j^{(4)} + s_i^{(4)}s_j^{(4)})   \\
    &\leq \frac{24C|\Omega|^3}{|\Omega|^4} \leq \frac{C'}{|\Omega|}
\end{align}
where $C$ is a constant

\section{Norm of weight vectors of quantum models}
\label{app:weingarten}

We consider a circuit with a diagonal encoding layer $S(x)$ applied to the $|+\rangle^n$ state followed by a trainable unitary $V$ and an observable $O$ such that $\text{Tr}(O) = 0$. We illustrate this example in \autoref{fig:Potential_Advantage_Framework}. 
The quantum model can then be written
\begin{equation}
    f_Q(x) = \text{Tr}(O\: VS(x)(|+\rangle\langle +|^n)S(x)^\dagger V^\dagger)
\end{equation}

The spectrum only depends on the encoding unitary $S(x)$, and each frequency $\omega$ has a multiplicity noted $R(\omega)$. We consider two types of encodings
among many possibilities that are detailed in \cite{peters2022generalization}:

\begin{itemize}
    \item \textbf{The ternary encoding}
    $S(x) = \bigotimes_{k=0}^{n-1} RZ_{k}(x\:3^k/2)$
    where $RZ_{k}$ denotes a $Z$ rotation applied to the qubit $k$. The spectrum produced by this encoding is the interval $\llbracket -(3^n-1)/2, \:\: (3^n-1)/2\rrbracket.$
    The multiplicity of a frequency $\omega$ depends on its ternary representation. Let us write $\omega$ in a ternary representation such that $\omega = \frac{1}{2}\sum_{k=0}^n t_k 3^k$ with $t_k \in \{-1, 0, 1\}$. The multiplicity depends on the number of zeros in the ternary representation of $\omega$ that we note $t(\omega)$. We have then $R(\omega) = 2^{n - t(\omega)}$.
    
    \item \textbf{The Golomb encoding}
    $S(x) = \exp(-i \displaystyle\frac{x}{2} R_G)$ where $R_G$ is a Golomb ruler \cite{piccard1939ensembles}. The resulting spectrum are all the integers in the set $\llbracket 0, N(N-1)/2\rrbracket.$ The frequency $0$ is of multiplicity $N$ and all other frequencies are of multiplicity $1$.
\end{itemize}

\begin{restatable}[]{thm}{WeightVectoretwodesign}\label{thm:EncodingsWeightVectorNorm}
    For the Golomb encoding, $\norm{\beta_Q}^2 \sim N$.
    For the ternary encoding, $\norm{\beta_Q}^2 \sim (3/2)^n$ 
\end{restatable}
\begin{proof}
We apply \cref{thm:WeightVectoretwodesign} to the two encodings.
    \begin{itemize}
        \item For the Golomb encoding, $p = N(N-1)/2$ and $R(\omega) = 1$ for all non zero frequencies. Therefore $\EE_V[\norm{\beta_Q}^2] = \frac{N(N-1)}{2(N+1)} \sim N$ and $\VV_V[\norm{\beta_Q}^2] = O(\frac{N^4}{N^6}N^2 + \frac{N^4}{N^4}) = O(1)$ . So with high probability $\norm{\beta_Q}^2 \sim N$.
        \item For the ternary encoding, 
$\EE_V[\norm{\beta_Q}^2] = \frac{3^n -1}{2^n} \sim (3/2)^n$
        \begin{equation}
            \sum_{\omega \in \Omega\backslash\{0\}}R(\omega)^2 = \sum_{k=0}^{n-1} (2^{k})^2 \binom{n}{k} = (4 + 1)^n - 4^n = N^2((5/4)^n - 1)
        \end{equation}
        So $\VV_V[\norm{\beta_Q}^2] = O((\frac{3}{4})^{2n}(\frac{5}{4})^{n} + (\frac{3}{4})^{2n}) = O(1)$
    \end{itemize}
\end{proof}

\begin{restatable}[]{thm}{WeightVectoretwodesign}\label{thm:WeightVectoretwodesign}
Let $f_Q(x) = \text{Tr}(O\: VS(x)(|+\rangle\langle +|^n)S(x)^\dagger V^\dagger)$ with $\text{Tr}(O)=0$.
    Let $V$ be drawn from a 2 design. We have that
    \begin{equation}
        \EE_V[\norm{\beta_Q}^2] = \frac{p}{N+1}
    \end{equation}
    Furthermore, if $V$ is drawn from a 4 design, we have that
    \begin{equation}
        \VV_V[\norm{\beta_Q}^2] = \Theta(\frac{p^2}{N^6}\sum_{\omega \in \Omega\backslash\{0\}}R(\omega)^2 + \frac{p^2}{N^4})
    \end{equation}
\end{restatable}

\begin{proof}

We first recall that: 

\begin{equation}
    \beta(0) = \frac{\sqrt{p}}{N}
    \sum^{R(0)}_{j=1}(V^{\dagger}OV)_{m_j,n_j} 
\end{equation}

Therefore the norm of $\beta$ can be computed as

\begin{align*}
    &\norm{\beta}^2 = |\beta(0)|^2 + \sum_{\omega \in \Omega} |\beta_{\cos}(\omega)|^2 + |\beta_{\sin}(\omega)|^2 \\
    &= \frac{p}{N^2} \bigg|\sum^{R(0)}_{j=1}(V^{\dagger}OV)_{m_j,n_j}\bigg|^2+ \frac{p}{N^2}\sum_{\omega\in \Omega}
    \bigg|\sum_{i=1}^{R(\omega)} [(V^{\dagger}OV)_{m^{(\omega)}_i,n^{(\omega)}_i} +i(V^{\dagger}OV)_{m^{(\omega)}_i,n^{(\omega)}_i}^*]\bigg|^2+\\
    &\bigg|\sum_{i=1}^{R(\omega)} [(V^{\dagger}OV)_{m^{(\omega)}_i,n^{(\omega)}_i} -i(V^{\dagger}OV)_{m^{(\omega)}_i,n^{(\omega)}_i}^*]\bigg|^2\\
    &= \frac{p}{N^2} \bigg|\sum^{R(0)}_{j=1}(V^{\dagger}OV)_{m_j,n_j}\bigg|^2+ 2\frac{p+1}{N^2}\sum_{\omega\in \Omega}
    \bigg|\sum_{i=1}^{R(\omega)} (V^{\dagger}OV)_{m^{(\omega)}_i,n^{(\omega)}_i}\bigg|^2\\
\end{align*}

We evaluate the moments of $\norm{\beta}^2$ with  the Weingarten calculus \cite{collins2006integration, cerezo2021cost}. \autoref{thm:2momentsHaar} and \autoref{thm:4momentsHaar} summarize intermediary results that will be used later in the proof. For better clarity of the reader, we first prove \autoref{thm:2momentsHaar} and \autoref{thm:4momentsHaar} for diagonal observables in \cref{subsec:diagonalObservable}, and we prove it for a general Pauli String in \cref{subsec:generalObservable}.

From \autoref{thm:2momentsHaar}
\begin{align}
    \EE[\norm{\beta}^2] &= \frac{p}{N^2}\EE[|v_0|^2] + 2\frac{p}{N^2}\sum_{\omega \in \Omega^*} \EE[|v_\omega|^2] = \frac{p}{N^2}\frac{N}{N^2-1} [R(0) + 2\sum_{\omega \in \Omega^*}R(\omega)]\\
    &=\frac{p}{N}\frac{N(N-1)}{N^2-1} = \frac{p}{N+1}
\end{align}

From \autoref{thm:4momentsHaar}

\begin{align}
    \EE[\norm{\beta}^4] &= \Big(\frac{p}{N^2} (|v_0|^2 + 2\sum_{\omega \in \Omega^*} |v_\omega|^2)\Big)^2\\
    &= \frac{p^2}{N^4}\Big(\EE[|v_0|^4| + 2\sum_{\omega \in \Omega^*}\EE[|v_\omega|^2|v_0|^2] +  4\sum_{\omega \in \Omega^*}\EE[|v_\omega|^4] + 4\sum_{\omega \neq \omega' \in \Omega^*}\EE[|v_\omega|^2|v_{\omega'}|^2]\Big)\\
    &= \frac{p^2}{N^4}\Big(\big[3R(0)^2 + 4\sum_{\omega \in \Omega^*}2R(\omega)^2\big]N^2W_4(N) + \mathcal{O}\big[R(0) + 2\sum_{\omega \in \Omega^*}R(\omega)\big]N^2W_2(N) \\
    &+ \big[2R(0) \sum_{\omega \in \Omega^*}R(\omega) + 4\sum_{\omega \neq \omega' \in \Omega^*}R(\omega)R(\omega')\big]N^2W_4(N) \\
    &+ \mathcal{O}\big[2R(0) \sum_{\omega \in \Omega^*}R(\omega) + 4\sum_{\omega \neq \omega' \in \Omega^*}R(\omega)R(\omega')\big]N^2W_2(N)\big)\\
    &= \frac{p^2}{N^2}\Big(\big[3R(0)^2 + 8\sum_{\omega \in \Omega^*}R(\omega)^2 + 2R(0) \sum_{\omega \in \Omega^*}R(\omega) + 4\sum_{\omega \neq \omega' \in \Omega^*}R(\omega)R(\omega')\big]W_4(N)\\
    &+ \mathcal{O}\big[R(0) + 2\sum_{\omega \in \Omega^*}R(\omega)+2R(0) \sum_{\omega \in \Omega^*}R(\omega) + 4\sum_{\omega \neq \omega' \in \Omega^*}R(\omega)R(\omega')\big]W_2(N)\big)\\
    &= \frac{p^2}{N^2}\Big(\big[(R(0) + 2\sum_{\omega \in \Omega^*}R(\omega))^2 + 2R(0)^2 + 4\sum_{\omega \in \Omega^*}R(\omega)^2 - 2R(0) \sum_{\omega \in \Omega^*}R(\omega)]W_4(N)\\
    &+ \mathcal{O}\big[R(0) + 2\sum_{\omega \in \Omega^*}R(\omega)(1+R(0)+ 2\sum_{\omega  \in \Omega^*}R(\omega))\big]W_2(N)\big)\\
    &= \frac{p^2}{N^2}\Big(\big[N^2(N-1)^2 + 2R(0)^2 + 4\sum_{\omega \in \Omega^*}R(\omega)^2 - 2R(0) \sum_{\omega \in \Omega^*}R(\omega)]W_4(N)\\
    &+ \mathcal{O}\big[R(0) + N^2\big]W_2(N)\big)
\end{align}

\begin{align}
    \EE[\norm{\beta}^4] - \EE[\norm{\beta}^2]^2=&\frac{p^2}{N^2}\Big(\big[N^2(N-1)^2 + 2R(0)^2 + 4\sum_{\omega \in \Omega^*}R(\omega)^2 - 2R(0) \sum_{\omega \in \Omega^*}R(\omega)]W_4(N)\\
    &+ \mathcal{O}\big[R(0) + N^2\big]W_2(N)\big) - \frac{p^2}{(N+1)^2}
\end{align}
\begin{align}
    \frac{p^2}{N^2}N^2(N-1)^2W_4(N) - \frac{p^2}{(N+1)^2} &= \frac{p^2(N-1)^2(N^4 - 8N^2+6)}{N^8 - 14N^6 +49N^4 - 36N^2} - \frac{p^2}{(N+1)^2} \\
    &= \frac{p^2(N-1)^2(N^4 - 8N^2+6)}{N^2(N^2-1)(N^4-13N^2+36))} - \frac{p^2}{(N+1)^2}\\
    &= \frac{p^2(N-1)(N^4 - 8N^2+6)}{N^2(N+1)(N^4-13N^2+36)} - \frac{p^2}{(N+1)^2}\\
    &= \frac{p^2(N-1)(N+1)(N^4 - 8N^2+6) - p^2N^2(N^4-13N^2+36))}{N^2(N+1)^2(N^4-13N^2+36))}\\
    &= \frac{4p^2N^4 + \mathcal{O}(p^2N^2)}{(N+1)^2N^2(N^4-13N^2+36))}\\
    &= \Theta(\frac{p^2}{N^4})\\
    W_4(N)&= \Theta(\frac{1}{N^4})\\
    N^2W_2(N)&= \mathcal{O}(\frac{1}{N^3})
\end{align}

\begin{align}
    \EE[\norm{\beta}^4] - \EE[\norm{\beta}^2]^2=\Theta(\frac{p^2}{N^6}\sum_{\omega \in \Omega}R(\omega)^2 ) + \Theta(\frac{p^2}{N^4})
\end{align}

\end{proof}

\subsection{Weingarten Integration for Diagonal obsrevable}
\label{subsec:diagonalObservable}

We prove the following result:
\begin{restatable}[]{thm}{}\label{thm:2momentsHaar}
    Let $U$ be a Haar random unitary and $O$ be a Pauli string composed only of $I$ and $Z$, different from identity. 
    Let
    \begin{equation}
    u = \sum_{i=1}^q (UOU^\dagger)_{m_in_i}, \quad \forall i, m_i \neq n_i \quad \text{and} \quad \forall i\neq j, \: (m_i, n_i) \neq (m_j, n_j)
    \end{equation} Then we have
    \begin{equation}
        \EE_{U\sim \text{Haar}}(|u|^2) = \frac{qN}{N^2 - 1}
    \end{equation}
\end{restatable}

\begin{proof}
    We have 
\begin{align}
    |u|^2 &= \Sum_{i, j}^q (UOU^\dagger)_{m_in_i} (UOU^\dagger)_{m_jn_j}^*
    =\Sum_{i, j=1}^q \big(\Sum_{k=1}^N u_{m_ik} O_k u_{n_ik}^*\big)\:\big(\Sum_{k=1}^N u_{m_jk} O_k u_{n_jk}^*\big)^* \\
    &= \Sum_{i, j=1}^q \Sum_{k, \ell =1}^N u_{m_ik}u_{n_j\ell}u_{n_ik}^*u_{m_j\ell}^* \; O_k O_\ell\\
\end{align}
With the Weingarten calculus we have \cite{collins2006integration, cerezo2021cost}: 
\begin{align}
\EE(u_{m_ik}u_{n_j\ell}u_{n_ik}^*u_{m_j\ell}^*) &= \frac{1}{N^2 - 1} [\delta_{m_in_i}\delta_{m_jn_j}\delta_{kk}\delta_{\ell\ell} + \delta_{m_im_j}\delta_{n_in_j}\delta_{k\ell}^2] \\
    &- \frac{1}{N(N^2-1)} [\delta_{m_in_i}\delta_{m_jn_j}\delta_{kl}^2 + \delta_{m_im_j}\delta_{n_in_j}\delta_{kk}\delta_{\ell\ell}]
\end{align}

 $\Sum_{k, \ell =1}^N O_kO_l = 0$ therefore we can eliminate all the terms in $\delta_{kk}\delta_{\ell\ell}$. Furthermore $\forall i \;\delta_{m_in_i} = 0$ and we have
 
\begin{align}
    \EE(|u|^2) &= \frac{1}{N^2 - 1}\Sum_{i, j=1}^q \Sum_{k =1}^N \delta_{m_im_j}\delta_{n_in_j}O_k^2 = \frac{N}{N^2 - 1}\Sum_{i, j=1}^q\delta_{m_im_j}\delta_{n_in_j}
\end{align}
We also made the hypothesis that $\forall i\neq j, \: (m_i, n_i) \neq (m_j, n_j)$ therefore $\delta_{m_im_j}\delta_{n_in_j} = 1$ iif $m_i = m_j$ and $n_i = n_j$ iif $i=j$.
Finally we have 
\begin{align}
    \EE(|u|^2) & = \frac{N}{N^2 - 1}\Sum_{i=1}^q\delta_{m_im_i}\delta_{n_in_i} = \frac{qN}{N^2 - 1} 
\end{align}
\end{proof}

We will prove the following result
\begin{restatable}[]{thm}{}
\label{thm:4momentsHaar}
    Let $u = \sum_{i=1}^q (UOU^\dagger)_{m_in_i}, u' = \sum_{i=1}^{q'} (UOU^\dagger)_{m'_in'_i}$. If $U$ is sampled from a $4$-design then
    \begin{equation}
        \EE_{U\sim \text{Haar}}[|u|^4] = 2q^2W_4(N)N^2 + \mathcal{O}(\displaystyle\frac{q}{N^3})
    \end{equation}
    If $\forall i_1$, there exists $i_2$ such that $(n_{i2}, m_{i2}) = (m_{i1}, n_{i1})$, then
    \begin{equation}
        \EE_{U\sim \text{Haar}}[|u|^4] = 3q^2W_4(N)N^2 + \mathcal{O}(\displaystyle\frac{q}{N^3})
    \end{equation}
    \begin{equation}
        \EE_{U\sim \text{Haar}}[|u|^2|u'|^2] = qq'W_4(N)N^2 + \mathcal{O}(\displaystyle\frac{qq'}{N^3})
    \end{equation}
    where we note
    \begin{equation}
    W_4(N) = \frac{N^4 - 8N^2+6}{N^8 - 14N^6 +49N^4 - 36N^2}, \quad W_2(N) = -\frac{1}{N^5 - 14N^3 + 9N}
\end{equation}
\end{restatable}

\begin{proof}

We have
\begin{align}
    |u|^4 &=  \sum_{i1, i2, i3, i4} (UOU^{\dagger})_{m_{i1},n_{i1}} (UOU^{\dagger})^{\ast}_{m_{i2},n_{i2}}
        (UOU^{\dagger})_{m_{i3},n_{i3}} (UOU^{\dagger})^{\ast}_{m_{i4},n_{i4}}\\
        & = \sum_{i1, i2, i3, i4} \sum_{k_1, k_2, k_3, k_4} u_{m_{i1},k_1} O_{k_1} u^{\ast}_{n_{i1},k_1} 
        u_{m_{i2},k_2}^{\ast} O_{k_2} u_{n_{i2},k_2}
        u_{m_{i3},k_3} O_{k_3} u^{\ast}_{n_{i3},k_3}
        u_{m_{i4},s}^{\ast} O_{k_4} u_{n_{i4},k_4}\\
        & = \sum_{i1, i2, i3, i4} \sum_{k_1, k_2, k_3, k_4} u_{m_{i1},k_1}u_{n_{i2},k_2}u_{m_{i3},k_3}u_{n_{i4},s}
u^{\ast}_{n_{i1},k_1}u^{\ast}_{m_{i2},k_2}u^{\ast}_{n_{i3},k_3}u^{\ast}_{m_{i4},k_4}
O_{k_1} O_{k_2}O_{k_3}O_{k_4}
\end{align}

By using the Weingarten calculus \cite{collins2006integration}, we have that 
\begin{align}
    &\EE(u_{m_{i1},k_1}u_{n_{i2},k_2}u_{m_{i3},k_3}u_{n_{i4},s}
u^{\ast}_{n_{i1},k_1}u^{\ast}_{m_{i2},k_2}u^{\ast}_{n_{i3},k_3}u^{\ast}_{m_{i4},k_4}) \\
&= \sum_{\sigma, \tau \in S_4} \delta_{m_{i1}, \sigma(m_{i1})} \cdot \delta_{m_{i3}, \sigma(m_{i3})} 
         \delta_{n_{i2}, \sigma(n_{i2})} \delta_{n_{i4}, \sigma(n_{i4})} \delta_{k_1, k_{\tau(1)}} \delta_{k_2, k_{\tau(2)}} \delta_{k_3, k_{\tau(3)}} \delta_{k_4, k_{\tau(4)}} W_g(\tau \sigma^{-1},N)
\end{align}
where $W_g$ is the Weingarten function and $S_4$ is the set of permutations on $\{1, 2, 3, 4\}$.

\begin{align}
    \EE[|u|^4] &= \sum_{i1, i2, i3, i4} \sum_{k_1, k_2, k_3, k_4} \sum_{\sigma, \tau \in S_4} O_{k_1} O_{k_2}O_{k_3}O_{k_4}\\
    &\delta_{m_{i1}, \sigma(m_{i1})}  \delta_{m_{i3}, \sigma(m_{i3})} 
         \delta_{n_{i2}, \sigma(n_{i2})} \delta_{n_{i4}, \sigma(n_{i4})} \delta_{k_1, k_{\tau(1)}} \delta_{k_2, k_{\tau(2)}} \delta_{k_3, k_{\tau(3)}} \delta_{k_4, k_{\tau(4)}} W_g(\tau \sigma^{-1},N)\\
    &=  \sum_{i1, i2, i3, i4} \sum_{\sigma \in S_4} \delta_{m_{i1}, \sigma(m_{i1})}  \delta_{m_{i3}, \sigma(m_{i3})} 
         \delta_{n_{i2}, \sigma(n_{i2})} \delta_{n_{i4}, \sigma(n_{i4})}\\
    &\sum_{\tau \in S_4} W_g(\tau \sigma^{-1},N)
    \sum_{k_1, k_2, k_3, k_4} \delta_{k_1, k_{\tau(1)}} \delta_{k_2, k_{\tau(2)}} \delta_{k_3, k_{\tau(3)}} \delta_{k_4, k_{\tau(4)}} O_{k_1} O_{k_2}O_{k_3}O_{k_4}\\
    &=\sum_{i1, i2, i3, i4} \sum_{\sigma, \tau \in S_4} \delta_{m_{i1}, \sigma(m_{i1})}  \delta_{m_{i3}, \sigma(m_{i3})} 
         \delta_{n_{i2}, \sigma(n_{i2})} \delta_{n_{i4}, \sigma(n_{i4})} A(\tau) W_g(\tau\sigma^{-1}, N)\\
\end{align}
where  we note $A(\tau) =  \Sum_{k_1, k_2, k_3, k_4} \delta_{k_1, k_{\tau(1)}} \delta_{k_2, k_{\tau(2)}} \delta_{k_3, k_{\tau(3)}} \delta_{k_4, k_{\tau(4)}} O_{k_1} O_{k_2}O_{k_3}O_{k_4}$.

Let us look at the term $A(\tau)$ for a given permutation $\tau$. We can verify that only for the permutations included in \autoref{tab:tau_perm} the sum will be non zero. For instance, 
 $A(\text{id}) = \sum_{k_1, k_2, k_3, k_4}  O_{k_1}O_{k_2}O_{k_3}O_{k_4} = \text{Tr}(O)^4 = 0$, and $A(\tau_1) = \sum_{k_1, k_2, k_3, k_4} \delta_{k_1, k_2} \delta_{k_2, k_1} \delta_{k_3, k_4} \delta_{k_4, k_3}  O_{k_1}O_{k_2}O_{k_3}O_{k_4} = \sum_{k_1, k_3} O_{k_1}^2 O_{k_3}^2 = N^2$. We also have $A(\tau_1) = A(\tau_2) = A(\tau_3) = N^2$

 So we have 
 \begin{align}
    \EE[|u|^4] &= \sum_{\sigma \in S_4}\sum_{i1, i2, i3, i4}  \delta_{m_{i1}, \sigma(m_{i1})}  \delta_{m_{i3}, \sigma(m_{i3})} 
         \delta_{n_{i2}, \sigma(n_{i2})} \delta_{n_{i4}, \sigma(n_{i4})} \sum_{\tau \in \{\tau_1, \tau_2, \tau_3\}}W_g(\tau \sigma^{-1},N) N^2
\end{align}

Let us look at the term $\sum_{i1, i2, i3, i4}  \delta_{m_{i1}, \sigma(m_{i1})}  \delta_{m_{i3}, \sigma(m_{i3})} 
         \delta_{n_{i2}, \sigma(n_{i2})} \delta_{n_{i4}, \sigma(n_{i4})}$ for a given permutation $\sigma$.

Since we assumed all elements $(UOU^\dagger)_{m_in_i}$ are off-diagonal, we have that $\delta_{m_i, n_i} = 0$, and it restricts the permutations potentially giving non-zero terms. These permutations are enumerated in \autoref{tab:sigma_perm}.

For each permutation $\sigma$ we compute the dominant term in $\sum_\tau A(\tau)W_g(\tau\sigma^{-1}, N)$. These numbers are enumerated in \autoref{tab:sigma_perm_cross}. We compute the number of non zero terms in $\sum_{i1, i2, i3, i4}  \delta_{m_{i1}, \sigma(m_{i1})}  \delta_{m_{i3}, \sigma(m_{i3})} 
         \delta_{n_{i2}, \sigma(n_{i2})} \delta_{n_{i4}, \sigma(n_{i4})}$. These numbers are enumerated in \autoref{tab:sigma_perm}

From \cite{weingarten4}, if $\sigma \neq \text{id}$, $W_g(\sigma, N) = \mathcal{O}(\frac{1}{N^5})$ and $W_g(\text{id}, N) = \displaystyle\frac{N^4-8N^2+6}{N^8-14N^6+49N^4-36N^2} =\mathcal{O}(\frac{1}{N^4})$. We can verify that only for $\sigma_1, \sigma_4, \sigma_7$ there exist $\tau \in {\tau_1, \tau_2, \tau_3}$ such that $\tau\sigma^{-1} = \text{id}$, and $\tau$ is unique.

By combining the numbers in \autoref{tab:sigma_perm} and \autoref{tab:sigma_tau} we have the final result.

\begin{table}[h!]
    \centering
\begin{tabular}{|c|cccc|cccc|c|}
 \hline
 & $k_1$ & $\sigma_O(k_2)$ & $k_3$ & $\sigma_O(k_4)$ &1&2&3&4& $A(\tau)$\\
 \hline
 $\tau_1$& $k_2$ & $\sigma_O(k_1)$ & $k_4$ & $\sigma_O(k_3)$ &2&1&4&3& $N^2$\\
 $\tau_2$& $k_4$ & $\sigma_O(k_3)$ & $k_2$ & $\sigma_O(k_1)$ &4&3&2&1& $N^2$\\
 $\tau_3$& $\sigma_O(k_3)$ &  $k_4$ & $\sigma_O(k_1)$ & $k_2$ & 3&4&1&2& $N^2$\\
 $\tau_4$& $k_2$ & $k_4$ & $\sigma_O(k_1)$ & $\sigma_O(k_3)$ &2&4&1&3& $N$\\
 $\tau_5$& $k_2$ & $\sigma_O(k_3)$ & $k_4$ & $\sigma_O(k_1)$ &2&3&4&1& $N$\\
 $\tau_6$& $k_4$ & $\sigma_O(k_1)$ & $k_2$ & $\sigma_O(k_3)$ &4&1&2&3& $N$\\
 $\tau_7$& $k_4$ & $\sigma_O(k_3)$ & $\sigma_O(k_1)$ & $k_2$ &4&3&1&2& $N$\\
 $\tau_8$& $\sigma_O(k_3)$ & $k_4$ & $k_2$ & $\sigma_O(k_1)$ & 3&4&2&1& $N$\\
 $\tau_9$& $\sigma_O(k_3)$ & $k_1$ & $k_4$ & $k_2$ & 3&1&4&2& $N$\\
 \hline
\end{tabular}
    \caption{$\tau$ permutations that give non zero values for non diagonal O. For diagonal $O$, consider $\sigma_O = id$.}
    \label{tab:tau_perm}
\end{table}

\begin{table}[h!]
    \centering
\begin{tabular}{|c|cccc|c|}
 \hline
 &$m_{i_1}$ & $n_{i_2}$ & $m_{i_3}$ & $n_{i_4}$& Number of non-zero terms\\
 \hline
 $\sigma_1$ & $m_{i_2}$ & $n_{i_1}$ & $m_{i_4}$ & $n_{i_3}$& $q^2$\\
 $\sigma_2$ & $m_{i_2}$ & $n_{i_3}$ & $m_{i_4}$ & $n_{i_1}$&$q+\mathcal{O}(q)$\\
 $\sigma_3$ & $m_{i_2}$ & $m_{i_4}$ & $n_{i_1}$ & $n_{i_3}$& $\mathcal{O}(q)$\\
 $\sigma_4$ & $m_{i_4}$ & $n_{i_3}$ & $m_{i_2}$ & $n_{i_1}$& $q^2$\\
 $\sigma_5$ & $m_{i_4}$ & $n_{i_1}$ & $m_{i_2}$ & $n_{i_3}$&$q+\mathcal{O}(q)$\\
 $\sigma_6$ & $m_{i_4}$ & $n_{i_3}$ & $n_{i_1}$ & $m_{i_2}$&$\mathcal{O}(q)$\\
 $\sigma_7$ & $n_{i_3}$ & $m_{i_4}$ & $n_{i_1}$ & $m_{i_2}$&$0$ or $q^2$\\
 \hline
\end{tabular}
    \caption{$\sigma$ permutations to be considered}
    \label{tab:sigma_perm}
\end{table}

Let us now look at
\begin{align}
    u_1 &= \sum_{i=1}^{p_1} (UOU^\dagger)_{m_i^{(1)}n_i^{(1)}}, \quad \forall i, m_i^{(1)} \neq n_i^{(1)} \quad \text{and} \quad \forall i\neq j, \: (m_i^{(1)}, n_i^{(1)}) \neq (m_j^{(1)}, n_j^{(1)})\\
    u_2 &= \sum_{i=1}^{p_2} (UOU^\dagger)_{m_i^{(2)}n_i^{(2)}}, \quad \forall i, m_i^{(2)} \neq n_i^{(2)} \quad \text{and} \quad \forall i\neq j, \: (m_i^{(2)}, n_i^{(2)}) \neq (m_j^{(2)}, n_j^{(1)})
\end{align}

\begin{align}
    |u_1|^2|u_2|^2 &= \bigg(\Sum_{i, j=1}^{p_1} \Sum_{k, \ell =1}^N u_{m_i^{(1)}k}u_{n_j^{(1)}\ell}u_{n_i^{(1)}k}^*u_{m_j^{(1)}\ell}^* \; O_k O_\ell\bigg)\bigg(\Sum_{i, j=1}^{p_2} \Sum_{k, \ell =1}^N u_{m_i^{(2)}k}u_{n_j^{(2)}\ell}u_{n_i^{(2)}k}^*u_{m_j^{(2)}\ell}^* \; O_k O_\ell\bigg)\\
    & = \sum_{i1, i2, i3, i4} \sum_{k_1, k_2, k_3, k_4} u_{m_{i1}^{(1)},k_1}u_{n_{i2}^{(1)},k_2}u_{m_{i3}^{(2)},k_3}u_{n_{i4}^{(2)},s}
u^{\ast}_{n_{i1}^{(1)},k_1}u^{\ast}_{m_{i2}^{(1)},k_2}u^{\ast}_{n_{i3}^{(2)},k_3}u^{\ast}_{m_{i4}^{(2)},k_4}
O_{k_1} O_{k_2}O_{k_3}O_{k_4}\\
\end{align}

$\EE[|u|^4] = 2q^2(W_4(N)N^2 + 4NW_2(N)) + \Theta(q)(2W_2(N)N^2 + NW_4(N)) = 2q^2W_4(N)N^2 + \mathcal{O}(\displaystyle\frac{q}{N^3})$

We do the same reasoning as above, but this time the number of non zero term for the permutations $\sigma$ are given \autoref{tab:sigma_perm_cross}.

\end{proof}

\begin{table}[h!]
    \centering
\begin{tabular}{|c|cccc|c|}
 \hline
 &$m_{i_1}^{(1)}$ & $n_{i_2}^{(1)}$ & $m_{i_3}^{(2)}$ & $n_{i_4}^{(2)}$& Number of non-zero terms \\
 \hline
 $\sigma_1$ & $m_{i_2}^{(1)}$ & $n_{i_1}^{(1)}$ & $m_{i_4}^{(2)}$ & $n_{i_3}^{(2)}$& $q_1q_2$ \\
 $\sigma_2$ & $m_{i_2}^{(1)}$ & $n_{i_3}^{(2)}$ & $m_{i_4}^{(2)}$ & $n_{i_1}^{(1)}$& $\mathcal{O}(q_1q_2)$\\
 $\sigma_3$ & $m_{i_2}^{(1)}$ & $m_{i_4}^{(2)}$ & $n_{i_1}^{(1)}$ & $n_{i_3}^{(2)}$& $\mathcal{O}(q_1q_2)$\\
 $\sigma_4$ & $m_{i_4}^{(2)}$ & $n_{i_3}^{(2)}$ & $m_{i_2}^{(1)}$ & $n_{i_1}^{(1)}$& 0\\
 $\sigma_5$ & $m_{i_4}^{(2)}$ & $n_{i_1}^{(1)}$ & $m_{i_2}^{(1)}$ & $n_{i_3}^{(2)}$& $\mathcal{O}(q_1q_2)$ \\
 $\sigma_6$ & $m_{i_4}^{(2)}$ & $n_{i_3}^{(2)}$ & $n_{i_1}^{(1)}$ & $m_{i_2}^{(1)}$& $\mathcal{O}(q_1q_2)$ \\
 $\sigma_7$ & $n_{i_3}^{(2)}$ & $m_{i_4}^{(2)}$ & $n_{i_1}^{(1)}$ & $m_{i_2}^{(1)}$&0\\
 \hline
\end{tabular}
    \caption{$\sigma$ permutations to be considered in the cross terms}
    \label{tab:sigma_perm_cross}
\end{table}

\begin{table}[h!]
    \centering
\begin{tabular}{|c|c|c|c|c|c|}
 \hline
 $\sigma$ & $n(\sigma, W_4, N^2)$ & $n(\sigma, W_2, N^2)$ & $n(\sigma, W_4, N)$ & $n(\sigma, W_2, N)$&$\sum_\tau A(\tau)W_g(\tau\sigma^{-1}, N)$ \\
 \hline
 $\sigma_1$ & 1 & 0&0&4&$N^2W_4(N) + 4NW_2(N)$\\
 $\sigma_2$  & 0&2&1&0&$NW_4(N) + 2N^2W_2(N)$\\
 $\sigma_3$ &0&2&1&0&$NW_4(N) + 2N^2W_2(N)$\\
 $\sigma_4$ & 1&0&0&4&$N^2W_4(N) + 4NW_2(N)$\\
 $\sigma_5$ &0&2&1&0&$NW_4(N) + 2N^2W_2(N)$\\
 $\sigma_6$ &0&2&1&0&$NW_4(N) + 2N^2W_2(N)$\\
 $\sigma_7$ &1&0&0&4&$N^2W_4(N) + 4NW_2(N)$\\
 \hline
\end{tabular}
    \caption{$\sigma$ permutations to be considered in the cross terms. For each permutation $\sigma$, we compute the dominant factors in the term $\sum_\tau A(\tau)W_g(\tau\sigma^{-1}, N)$. To do so, for each $\sigma$, we compute the number of permutations $\tau$ such that $A(\tau)W_g(\tau\sigma^{-1}, N)$ is equal to $N^2W_4(N), NW_4(N), N^2W_2(N), NW_2(N)$. We use a computer to find the numbers in the first four columns.}
    \label{tab:sigma_tau}
\end{table}

\subsection{Weingarten Integration for general Pauli observable}
\label{subsec:generalObservable}

We do the same computation as above, but we no longer assume that $O \in \{I, Z\}^{\otimes n}$. We indeed assume that $O$ is a general Pauli observable. Then there exists a permutation $\sigma_O \neq id$ such that for each row $k$ of $O$, $O_{k, \sigma_O(k)} \neq 0$. Because $O$ is hermitian we have $\sigma_O^2 = id$ And we have 

\begin{equation}
    (VOV^\dagger)_{ij} = \sum_{k=1}^n v_{ik}O_{k, \sigma_O(k)}v^*_{j, \sigma_O(k)}
\end{equation}

and 
\begin{align}
    |(VOV^\dagger)_{ij}|^2 &= \bigg(\Sum_{k=1}^N v_{ik} O_{k, \sigma_O(k)} v_{j, \sigma_O(k)}^*\bigg)\:\bigg(\Sum_{k=1}^N v_{ik} O_{k, \sigma_O(k)} v_{j, \sigma_O(k)}^*\bigg)^*\\
    &= \Sum_{k, \ell =1}^N v_{ik} O_{k, \sigma_O(k)} v_{j, \sigma_O(k)}^* \quad v_{i\ell}^* O_{\ell, \sigma_O(\ell)}^* v_{j, \sigma_O(\ell)}\\
    &= \Sum_{k, \ell =1}^N v_{ik}v_{j, \sigma_O(\ell)}   v_{j, \sigma_O(k)}^*v_{i\ell}^* O_{k, \sigma_O(k)}O_{\ell, \sigma_O(\ell)}^*
\end{align}

More precisely, we want to evaluate $\EE[\abs{v}^2]$ where 
\begin{equation}
    v = \sum_{i=1}^q (VOV^\dagger)_{m_in_i}, \quad \forall i, m_i \neq n_i \quad \text{and} \quad \forall i\neq j, \: (m_i, n_i) \neq (m_j, n_j)
\end{equation}
We have
\begin{align}
    \abs{v}^2 &= \Sum_{i, j}^q (VOV^\dagger)_{m_in_i} (VOV^\dagger)_{m_jn_j}^*\\
    &=\Sum_{i, j=1}^q \big(\Sum_{k=1}^N v_{m_ik} O_{k, \sigma_O(k)} v_{n_i, \sigma_O(k)}^*\big)\:\big(\Sum_{k=1}^N v_{m_jk} O_{k, \sigma_O(k)} v_{n_j \sigma_O(k)}^*\big)^* \\
    &= \Sum_{i, j=1}^q \Sum_{k, \ell =1}^N v_{m_ik}v_{n_j\sigma_O(\ell)}v_{n_i\sigma_O(k)}^*v_{m_j\ell}^* \; O_{k, \sigma_O(k)} O_{\ell, \sigma_O(\ell)}^*\\
    \EE(v_{m_ik}v_{n_j\sigma_O(\ell)}v_{n_i\sigma_O(k)}^*v_{m_j\ell}^*) &= \frac{1}{N^2 - 1} [\delta_{m_in_i}\delta_{m_jn_j}\delta_{k,\sigma_O(k)}\delta_{\ell,\sigma_O(l)} + \delta_{m_im_j}\delta_{n_in_j}\delta_{kl}\delta_{\sigma_O(l)\sigma_O(k)}] \\
    &- \frac{1}{N(N^2-1)} [\delta_{m_in_i}\delta_{m_jn_j}\delta_{kl}\delta_{\sigma_O(l)\sigma_O(k)} + \delta_{m_im_j}\delta_{n_in_j}\delta_{k,\sigma_O(k)}\delta_{\ell,\sigma_O(l)}]
\end{align}

We have that
\begin{equation}
   \forall k,l, \:\delta_{k,\sigma_O(k)}\delta_{\ell,\sigma_O(l)} = 0, \quad \Sum_{k, \ell =1}^N \delta_{kl}\delta_{\sigma_O(l)\sigma_O(k)} O_{k, \sigma_O(k)} O_{\ell, \sigma_O(\ell)} = \Sum_{k =1}^N O_{k, \sigma_O(k)}O_{k, \sigma_O(k)}^* = N
\end{equation}

\begin{align}
    \EE[\abs{v}^2] &= \frac{1}{N^2 - 1} \Sum_{i, j=1}^q \delta_{m_im_j}\delta_{n_in_j} N - \frac{1}{N(N^2-1)} \Sum_{i, j=1}^q \delta_{m_in_i}\delta_{m_jn_j} N\\
    &= \frac{qN}{N^2 - 1}
\end{align}

We note 
\begin{align}
    v_1 = \sum_{i=1}^q (VOV^\dagger)_{m_i^{(1)}n_i^{(1)}}\\
    v_2 = \sum_{i=1}^q (VOV^\dagger)_{m_i^{(2)}n_i^{(2)}}\\
\end{align}

We want to compute $\EE[v_1v_2^*]$ and $\EE[v_1v_2]$.

We have
\begin{align}
    v_1v_2^* &= \Sum_{i=1}^{q_1} \Sum_{j=1}^{q_2}(VOV^\dagger)_{m_i^{(1)}n_i^{(1)}} (VOV^\dagger)_{m_j^{(2)}n_j^{(2)}}^*\\
    &=\Sum_{i=1}^{q_1} \Sum_{j=1}^{q_2}\big(\Sum_{k=1}^N v_{m_i^{(1)}k} O_{k, \sigma_O(k)} v_{n_i^{(1)}, \sigma_O(k)}^*\big)\:\big(\Sum_{k=1}^N v_{m_j^{(2)}k} O_{k, \sigma_O(k)} v_{n_j^{(2)} \sigma_O(k)}^*\big)^* \\
    &= \Sum_{i=1}^{q_1} \Sum_{j=1}^{q_2} \Sum_{k, \ell =1}^N v_{m_i^{(1)}k}v_{n_j^{(2)}\sigma_O(\ell)}v_{n_i^{(1)}\sigma_O(k)}^*v_{m_j^{(2)}\ell}^* \; O_{k, \sigma_O(k)} O_{\ell, \sigma_O(\ell)}^*\\
    \EE(v_{m_i^{(1)}k}v_{n_j^{(2)}\sigma_O(\ell)}v_{n_i^{(1)}\sigma_O(k)}^*v_{m_j^{(2)}\ell}^*) &= \frac{1}{N^2 - 1} [\delta_{m_i^{(1)}n_i^{(1)}}\delta_{m_j^{(2)}n_j^{(2)}}\delta_{k,\sigma_O(k)}\delta_{\ell,\sigma_O(l)} + \delta_{m_i^{(1)}m_j^{(2)}}\delta_{n_i^{(1)}n_j^{(2)}}\delta_{kl}\delta_{\sigma_O(l)\sigma_O(k)}] \\
    &- \frac{1}{N(N^2-1)} [\delta_{m_i^{(1)}n_i^{(1)}}\delta_{m_j^{(2)}n_j^{(2)}}\delta_{kl}\delta_{\sigma_O(l)\sigma_O(k)} + \delta_{m_i^{(1)}m_j^{(2)}}\delta_{n_i^{(1)}n_j^{(2)}}\delta_{k,\sigma_O(k)}\delta_{\ell,\sigma_O(l)}]
\end{align}

\begin{align}
    v_1v_2 
    &=\Sum_{i=1}^{q_1} \Sum_{j=1}^{q_2}\big(\Sum_{k=1}^N v_{m_i^{(1)}k} O_{k, \sigma_O(k)} v_{n_i^{(1)}, \sigma_O(k)}^*\big)\:\big(\Sum_{k=1}^N v_{m_j^{(2)}k} O_{k, \sigma_O(k)} v_{n_j^{(2)} \sigma_O(k)}^*\big) \\
    &= \Sum_{i=1}^{q_1} \Sum_{j=1}^{q_2} \Sum_{k, \ell =1}^N v_{m_i^{(1)}k}v_{m_j^{(2)}\sigma_O(\ell)}v_{n_i^{(1)}\sigma_O(k)}^*v_{n_j^{(2)}\ell}^* \; O_{k, \sigma_O(k)} O_{\ell, \sigma_O(\ell)}^*\\
    \EE(v_{m_i^{(1)}k}v_{n_j^{(2)}\sigma_O(\ell)}v_{n_i^{(1)}\sigma_O(k)}^*v_{m_j^{(2)}\ell}^*) &= \frac{1}{N^2 - 1} [\delta_{m_i^{(1)}n_i^{(1)}}\delta_{m_j^{(2)}n_j^{(2)}}\delta_{k,\sigma_O(k)}\delta_{\ell,\sigma_O(l)} + \delta_{m_i^{(1)}n_j^{(2)}}\delta_{n_i^{(1)}m_j^{(2)}}\delta_{kl}\delta_{\sigma_O(l)\sigma_O(k)}] \\
    &- \frac{1}{N(N^2-1)} [\delta_{m_i^{(1)}n_i^{(1)}}\delta_{m_j^{(2)}n_j^{(2)}}\delta_{kl}\delta_{\sigma_O(l)\sigma_O(k)} + \delta_{m_i^{(1)}n_j^{(2)}}\delta_{n_i^{(1)}m_j^{(2)}}\delta_{k,\sigma_O(k)}\delta_{\ell,\sigma_O(l)}]
\end{align}

\begin{align}
    \EE[v_1v_2^*] &= \frac{1}{N^2 - 1} \Sum_{i=1}^{q_1} \Sum_{j=1}^{q_2} \delta_{m_i^{(1)}m_j^{(2)}}\delta_{n_i^{(1)}n_j^{(2)}} N - \frac{1}{N(N^2-1)} \Sum_{i=1}^{q_1} \Sum_{j=1}^{q_2}\delta_{m_i^{(1)}n_i^{(1)}}\delta_{m_j^{(2)}n_j^{(2)}} N = 0\\
    \EE[v_1v_2] &= \frac{1}{N^2 - 1} \Sum_{i=1}^{q_1} \Sum_{j=1}^{q_2} \delta_{m_i^{(1)}n_j^{(2)}}\delta_{n_i^{(1)}m_j^{(2)}} N - \frac{1}{N(N^2-1)} \Sum_{i=1}^{q_1} \Sum_{j=1}^{q_2}\delta_{m_i^{(1)}n_i^{(1)}}\delta_{m_j^{(2)}n_j^{(2)}} N\\
    &= 0
\end{align}

\begin{align}
    \abs{v}^4 &=  \sum_{i1, i2, i3, i4} (VOV^{\dagger})_{m_{i1},n_{i1}} (VOV^{\dagger})^{\ast}_{m_{i2},n_{i2}}
        (VOV^{\dagger})_{m_{i3},n_{i3}} (VOV^{\dagger})^{\ast}_{m_{i4},n_{i4}}\\
        & = \sum_{i1, i2, i3, i4} \sum_{k_1, k_2, k_3, k_4} \\
        &v_{m_{i1},k_1} O_{k_1, \sigma_O(k_1)} v^{\ast}_{n_{i1},\sigma_O(k_1)}
        v_{m_{i2},k_2}^{\ast} O_{k_2, \sigma_O(k_2)}^* v_{n_{i2},\sigma_O(k_2)}
        v_{m_{i3},k_3} O_{k_3, \sigma_O(k_3)} v^{\ast}_{n_{i3},\sigma_O(k_3)}
        v_{m_{i4},k_4}^{\ast} O_{k_4, \sigma_O(k_4)} v_{n_{i4},\sigma_O(k_4)}^*\\
        & = \sum_{i1, i2, i3, i4} \sum_{k_1, k_2, k_3, k_4}\\
        &v_{m_{i1},k_1}v_{n_{i2},\sigma_O(k_2)}v_{m_{i3},k_3}v_{n_{i4},\sigma_O(k_4)}v^{\ast}_{n_{i1},\sigma_O(k_1)} v_{m_{i2},k_2}^{\ast}v^{\ast}_{n_{i3},\sigma_O(k_3)}v_{m_{i4},k_4}^{\ast}O_{k_1, \sigma_O(k_1)} O_{k_2, \sigma_O(k_2)}^* O_{k_3, \sigma_O(k_3)} O_{k_4, \sigma_O(k_4)}^* 
\end{align}

\begin{align}
    \EE[\abs{v}^4] &= \sum_{i1, i2, i3, i4}  \sum_{\sigma, \tau \in S_4}\delta_{m_{i1}, \sigma(m_{i1})}  \delta_{m_{i3}, \sigma(m_{i3})} 
         \delta_{n_{i2}, \sigma(n_{i2})} \delta_{n_{i4}, \sigma(n_{i4})}  A(\tau) W_g(\tau\sigma^{-1}, N)\\ 
\end{align}

where $A(\tau) = \Sum_{k_1, k_2, k_3, k_4} \delta_{k_1, \tau(k_1)} \delta_{\sigma_O(k_2), \tau(\sigma_O(k_2))} \delta_{k_3, \tau(k_3)} \delta_{\sigma_O(k_4), \tau(\sigma_O(k_4))}O_{k_1, \sigma_O(k_1)}O_{k_2, \sigma_O(k_2)}^* O_{k_3, \sigma_O(k_3)} O_{k_4, \sigma_O(k_4)}^*$

We then do the same reasoning as with $O$ diagonal.

\section{Re-uploading Fourier models}
\label{app:Reuploading_Proof}

In this Section, we offer more formal versions of the Theorems presented in \autoref{subsec:ReUploadingModel} with their corresponding proofs. We start by \autoref{thm:Exp_BetaQ_Reuploading_2design}:

\begin{restatable}[Formal version of \autoref{thm:Exp_BetaQ_Reuploading_2design}]{thm}{ExpBetaQReuploadingtwodesignFormal}\label{thm:Exp_BetaQ_Reuploading_2design_Formal}
    Consider a single layered quantum re-uploading model with Fourier coefficients $c_{\omega}(\theta)$, with spectrum $\Omega$. We assume that each of the two parameterized unitaries are drawn form a 2-design. The variance of $\norm{\beta_Q}_2$ is given by:
    \begin{equation}
        \begin{split}
            \mathbb{E}_{\text{Haar}}[\norm{\beta_Q}^2_2] =  &\left(\frac{N\norm{O}_2^2-\text{Tr}(O)^2}{N(N^2-1)}\right) \frac{N^2p}{N(N+1)}\\
            &+ \frac{\text{Tr}(O)^2}{N^2} \, \textrm{,}
        \end{split}
    \end{equation}
    with $O$ the measurement observable.
\end{restatable}

We observe that, assuming $\text{Tr}(O)=0$ and $\norm{O}_2^2=N$, the expected values $\mathbb{E}_{\text{Haar}}[\norm{\beta_Q}^2_2]$ is scaling as $\frac{p}{N}$ such as in previous examples from \autoref{subsec:Limitations}. Using Jensen's inequality, we have that $\mathbb{E}_{\text{Haar}}[\norm{\beta_Q}_2] \leq \sqrt{\mathbb{E}_{\text{Haar}}[\norm{\beta_Q}^2_2]}$. Therefore, we have that the norm of $\beta_Q$ can be very low for low value of $p$, while the case $p \sim N$ may offer a potential advantage. Note that one could obtain an expression for $\VV[\norm{\beta_Q}^2_2]$ similarly as in \autoref{app:non-concnetration}, but under the hypothesis that the trainable layer unitaries form an $8$-design; this would require using Weingarten calculus of order $8$.

\begin{proof}
The expression of the quantum model weight vector l2-norm is $\norm{\beta_Q}_2 = \sqrt{\sum_{i=1}^{|\Omega|} |c_{\omega}|^2}$, thus we have:
    \begin{equation}
        \mathbb{E}[\norm{\beta_Q}_2^2] = \sum_{\omega \in \Omega} \mathbb{E}[|c_{\omega}|^2] \, \textrm{.}
    \end{equation}

\begin{restatable}[from \cite{mhiri2024constrained}]{thm}{ThmMhiri1}
\label{thm:single_layer_gloabl_2design}
     Consider a single layered Quantum Fourier model with Fourier coefficients $c_{\omega}(\theta)$, with spectrum $\Omega$, and redundancies $|R(\omega)|$. We assume that  each of the two parameterized layers form independently a 2-design (under the uniform distribution over their parameters). The expectation and variance of each Fourier coefficient in the spectrum $\Omega$ is given by:
\begin{equation}\label{eq:coeff_variance_2design_single}
    \begin{aligned}
        \mathbb{E}_{\theta}[c_{\omega}(\theta)]\quad&= \quad\frac{Tr(O)}{N}\delta_{\omega}^0 \, \textrm{,}\\
        \text{Var}_{\theta}[c_{\omega}(\theta)] &=	 \left(\frac{N\norm{O}_2^2-Tr(O)^2}{N(N^2-1)}\right) \frac{|R(\omega)|}{N(N+1)}    +   \frac{Tr(O)^2-N\norm{O}^2}{N^2(N^2-1)}\delta_{\omega}^0 \, \textrm{.}
    \end{aligned}
\end{equation}
\end{restatable}

According to the fact that $\sum_{\omega \in \Omega} |R(\omega)| = N^2 = 2^{2n}$, we now that each Fourier coefficient variance is vanishing when the trainable layers describe a 2-design. In addition, we observe that:
\begin{equation}
    \mathbb{E}_{\text{Haar}}[\norm{\beta_Q}^2_2] = \mathbb{E}_{\text{Haar}}[p\sum_{\omega \in \Omega} |c_{\omega}(\theta)|^2] = \sum_{\omega \in \Omega}p \mathbb{E}_{\text{Haar}}[|c_{\omega}(\theta)|^2] =  \sum_{\omega \in \Omega}p (\text{Var}_{\text{Haar}}[c_{\omega}(\theta)] + \mathbb{E}_{\text{Haar}}[c_{\omega}(\theta)]^2)
\end{equation}
And thus:
\begin{equation}
    \mathbb{E}_{\text{Haar}}[\norm{\beta_Q}^2_2] = \left(\frac{N\norm{O}_2^2-Tr(O)^2}{N(N^2-1)}\right) \frac{pN^2}{N(N+1)} + \frac{Tr(O)^2}{N^2}
\end{equation}
\end{proof}

In \cite{mhiri2024constrained}, the authors offer a bound on the variance of Fourier coefficients according to the monomial distance $\varepsilon$ of each trainable layer unitary matrix to a 2 design. Similarly, we offer a bound on the variance of the weight vector:

\begin{restatable}[]{thm}{bornbetaqhamiltonianencoding}
\label{thm:born_betaq_hamiltonian_encoding}
     Consider a single layered quantum re-uploading model with Fourier coefficients $c_{\omega}(\theta)$, with spectrum $\Omega$. We assume that each of the two parameterized unitaries form an $\varepsilon$-approximate 2-design according to the monomial definition. The expectation and variance of $\norm{\beta_Q}_2$ are given by:
\begin{equation}
    \begin{split}
        \mathbb{E}_{\theta}[\norm{\beta_Q}^2_2] \leq \mathbb{E}_{\text{Haar}}[\norm{\beta_Q}^2_2]&+\left( \frac{C_1 \varepsilon}{N^2}+\frac{C_2\varepsilon}{N(N+1)}\right)N^2 \\
        &+C_2\frac{\varepsilon^2}{N^2}N^4
    \end{split}
\end{equation}
where $C_1 = \Frac{N\norm{O}_2^2-Tr(O)^2}{N(N^2-1)}$,$C_2=\Sum_{l,k}\Frac{|O^{\bigotimes{2}}_{l,k}|}{N^2}$, and $\mathbb{E}_{\text{Haar}}[\norm{\beta_Q}^2_2]$ as defined in \autoref{thm:Exp_BetaQ_Reuploading_2design}.
\end{restatable}

Once again, one can use Jensen's inequality to derive a similar bound on $\mathbb{E}_{\theta}[\norm{\beta_Q}_2]$.  According to the trainable layers monomial distance $\varepsilon$ to a 2-design, the choice of the observable, and the dimension of the feature map $p$, the l2-norm of the quantum weight vector can be sufficiently low and close to $\norm{\bmnls}_2$ or very large (when $p \sim N^2$).

\begin{proof}
In a more general setting where trainable layers are $\varepsilon$-approximate 2-design according to the monomial distance, the authors from \cite{mhiri2024constrained} provide a bound on the Fourier coefficients variance:

\begin{restatable}[from \cite{mhiri2024constrained}]{thm}{ThmMhiri2}
\label{thm:bound_approx_2design}
        Consider a single layered Quantum Fourier model with spectrum $\Omega$, Fourier coefficients $c_{\omega}(\theta)$   and redundancies $|R(\omega)|$. We assume that each of the two parameterized layers forms an $\varepsilon$-approximate 2-design according to the monomial definition. The variance of the model's Fourier coefficients obeys the following bound:
    \begin{equation}
     \text{Var}_{\theta}[c_{\omega}(\theta)] \leq \text{Var}_{\text{Haar}}[c_{\omega}(\theta)]+\left( \frac{C_1 \varepsilon}{d^2}+\frac{C_2\varepsilon}{d(d+1)}\right)|R(\omega)|+C_2\frac{\varepsilon^2}{d^2}|R(\omega)|^2 \, \textrm{,}
   \end{equation}
where $C_1 = \frac{d\norm{O}^2-Tr(O)^2}{d(d^2-1)},C_2=\sum_{l,k} \frac{|O^{\bigotimes{2}}_{l,k}|}{d^2} $ and $\text{Var}_{\text{Haar}}[c_{\omega}]$ is the variance of a Fourier coefficient under the 2-design assumption given in \autoref{thm:single_layer_gloabl_2design}.
\end{restatable}

Considering a single layered Quantum Fourier model, and assuming that each of the two parametrized layers form independently a 2-design (under the uniform distribution over their parameters), we can use the results from \autoref{thm:single_layer_gloabl_2design} and \autoref{thm:bound_approx_2design}. By applying \autoref{thm:bound_approx_2design}, it comes directly:
    \begin{equation}
        \mathbb{E}_{\theta}[\norm{\beta_Q}^2_2] \leq \mathbb{E}_{\text{Haar}}[\norm{\beta_Q}^2_2]+\left( \frac{C_1 \varepsilon}{d^2}+\frac{C_2\varepsilon}{d(d+1)}\right)|R(\omega)|+C_2\frac{\varepsilon^2}{d^2}|R(\omega)|^2 \, \textrm{.}
    \end{equation}
    By using Jensen inequality, and by considering the concavity of the square root function:
    \begin{equation}
        \mathbb{E}_{\theta}[\norm{\beta_Q}_2] \leq \sqrt{\mathbb{E}_{\text{Haar}}[\norm{\beta_Q}^2_2]+\left( \frac{C_1 \varepsilon}{d^2}+\frac{C_2\varepsilon}{d(d+1)}\right)|R(\omega)|+C_2\frac{\varepsilon^2}{d^2}|R(\omega)|^2}
    \end{equation}
\end{proof}

\section{Discrete Logarithm}
\label{app:discreteLog}

In this section, we detail examples of quantum models based on the discrete logarithm primitive.

Let the discrete logarithm unitary be defined as
\begin{equation}
    U_{\text{DLP}}: |i\rangle \longmapsto |\log_g i + 1 \rangle
\end{equation}
where $g$ is a prime number in $\llbracket 0, N-1\rrbracket$.

Let $|\psi (x)\rangle = \bigotimes_{i=1}^n RY(x_i) |0^n\rangle$ and 
\begin{equation}
    f_{\text{DLP}}(x) = \text{Tr}(U_{\text{DLP}}^{\dagger}Z_0 U_{\text{DLP}}|\psi(x)\rangle\langle\psi(x)|)
\end{equation}

$U_{\text{DLP}}^{\dagger}Z_0 U_{\text{DLP}}$ is a hermitian diagonal matrix and the coefficients can be written as $(U_{\text{DLP}}^{\dagger}Z_0 U_{\text{DLP}})_{ii} = (-1)^{b_0(\log i + 1)}$. $b_0(j)$ is the 0th bit of the binary description of $j$.

One can then rewrite $f_{\text{DLP}}$ as 
\begin{align}
    f_{\text{DLP}}(x) &= \sum_y [\prod_{i=0}^{n-1} \cos (x_i/2)^{1-y_i}\sin (x_i/2)^{y_i}]^2 b_y = \sum_y \prod_{i=0}^{n-1} (\cos (x_i/2)^2)^{1-y_i}(1 - \cos (x_i/2)^2)^{y_i} b_y\\
    &= \sum_y \frac{1}{2^n} \prod_{i=0}^{n-1} (1 + \cos (x_i))^{1-y_i}(1 - \cos (x_i))^{y_i} b_y = \sum_y \frac{1}{2^n}b_y \prod_{i=0}^{n-1} (1 + (-1)^{y_i}\cos (x_i)) \\
    &= \sum_y b_y \: \phi_y(x)
\end{align}

The basis function in the feature space can be defined as 
\begin{equation}
    \bigg\{\phi_y (x) = \frac{1}{2^n}\prod_{i=0}^{n-1} (1 + (-1)^{y_i}\cos (x_i)) \: , \:y\:\in \{0, 1\}^n \bigg\}
\end{equation}

\begin{align}
    \EE_{x \sim \{0, \pi\}^n}[(\phi_y(x))^2] &= \frac{1}{2^n}
\end{align}

and therefore
\begin{align}
    |b_y \: 2^n\norm{|\phi_y(x)}_{\mu} = \frac{2^n}{\sqrt{2}^n} = \sqrt{2}^n
\end{align}

\section{Proof of \autoref{thm:Example_Perfect_Fct}}
\label{app:non-concnetration}
We recall \autoref{thm:Example_Perfect_Fct}:
\ThmFunctionFarMNLSNoConcentration*

\begin{proof}
    We first remark that for all $\omega$, $\EE[\beta_{\omega, \cos}^2] = \EE[\beta_{\omega, \sin}^2] = \sigma^2/3$.

    The first item is proven by applying the Hoeffding inequality to the random variable $\norm{\beta}^2 = \Sum_\omega \beta_{\omega, \cos}^2 + \beta_{\omega, \sin}^2$ since every element of the sum is iid, and $\EE[\norm{\beta}^2] = \frac{2}{3} p\sigma^2$.

    To prove the second item, we remind that for all $\omega \neq 0$,  $\EE_x[\cos(\omega^\top x)] = \EE_x[\sin(\omega^\top x)] = 0$ and $\EE_x[\cos(\omega^\top x)^2] = \EE_x[\sin(\omega^\top x)^2] = 1/2, \:\: \EE_x[\cos(\omega^\top x)\sin(\omega^\top x)] = 0$ and for all $\omega \neq \omega'$, $\EE_x[\cos(\omega^\top x)\sin(\omega'^\top x)] = 
    \EE_x[\cos(\omega^\top x)\cos(\omega'^\top x)] = 
    \EE_x[\sin(\omega^\top x)\sin(\omega'^\top x)] = 0$.
    
    Then we have $\EE_x[f(x)] = 0$ and $\EE_x[f(x)^2] = \frac{1}{p}\Sum_\omega  \beta_{\omega, \cos}^2 + \beta_{\omega, \cos}^2 = \frac{\norm{\beta}^2}{p}$ and we apply the result of the first item.

    We will now show that $|f(x)| \leq 1$ on $\mathbb{R}^d$.

First, we show that $\forall x, \PP(|f(x)| \geq \epsilon) \leq 2\exp(-\displaystyle\frac{\epsilon^2}{2\sigma^2})$.
Let $i_0 \in [1, p]$ and $\beta' = (\beta_{\omega_1, \cos}, \beta_{\omega_1, \sin} \dots \beta_{\omega_{i_0}, \cos}', \beta_{\omega_{i_0}, \sin},  \dots \beta_{\omega_p})$.
\begin{align}
    f_\beta(x) - f_{\beta'}(x) &= \frac{1}{\sqrt{p}}\Sum_{\omega \neq \omega_{i_0}}  (\beta_{\omega, \cos}\cos(\omega^\top x) + \beta_{\omega, \sin}\sin(\omega^\top x)) + \frac{\beta_{\omega_{i_0}, \cos}}{\sqrt{p}} \cos(\omega^\top x) + \frac{\beta_{\omega_{i_0}, \sin}}{\sqrt{p}}\sin(\omega^\top x)\\
    &- \frac{1}{\sqrt{p}}\Sum_{\omega \neq \omega_{i_0}}  (\beta_{\omega, \cos}\cos(\omega^\top x) + \beta_{\omega, \sin}\sin(\omega^\top x)) - \frac{\beta_{\omega_{i_0}, \cos}'}{\sqrt{p}} \cos(\omega^\top x) - \frac{\beta_{\omega_{i_0}, \sin}}{\sqrt{p}}\sin(\omega^\top x)\\
    |f_\beta(x) - f_{\beta'}(x)|&= \frac{1}{\sqrt{p}} \:|\cos(\omega^\top x)|\:|\beta_{\omega_{i_0}, \cos} - \beta'_{\omega_{i_0}, \cos}|\\
    &\leq \frac{2\sigma}{\sqrt{p}}
\end{align}
The last inequality comes from $|\beta_{\omega_{i_0}} - \beta'_{\omega_{i_0}}| \leq 2\sigma$ and $|\cos(\omega^\top x)| \leq 1$. The same result can be obtain with $\beta' = (\beta_{\omega_1, \cos}, \beta_{\omega_1, \sin} \dots \beta_{\omega_{i_0}, \cos}, \beta_{\omega_{i_0}, \sin}',  \dots \beta_{\omega_p})$. Furthermore, for all $x$, $\EE_\beta[f_\beta(x)] = 0$

From McDiarmid's inequality
\begin{equation}\label{eqn:mcdiarmid}
    \PP_{\beta}(|f_\beta(x)| \geq \epsilon) \leq 2\exp(-\frac{2\epsilon^2}{\sum_{i=1}^p \frac{4\sigma^2}{p}}) \leq 2\exp(-\displaystyle\frac{\epsilon^2}{2\sigma^2})
\end{equation}

Now we prove that $|f|\leq 1$ on $\RR^d$. $f$ is a periodic function, so we only need to prove that $f$ is bounded over the domain $\mathcal{X} = [0, 2\pi]^d$.
We consider a covering net of radius $\epsilon$ of $\mathcal{X}$. Let $T=\{r_1, \dots r_T\}$ be the set of centers of the balls composing the covering net.

Let $r = (r_1, \dots r_d)\in T$ and $u = (u_1, \dots u_d)\in \mathcal{X}$ such that $\norm{u} \leq \epsilon$.

Then we have, by applying Taylor's formulas
\begin{align}
    f(r+u) - f(r) &= \Sum_{i=1}^d\partial_i f(r)u_i + \Sum_{i, j} R_{ij}(r+u) u_iu_j\\
    &= u^\top \nabla f (r) + u^\top R(r, u) u\\
    \text{with} \quad [R(r,u)]_{ij} = R_{ij}(r,u) &\:= \:S_{ij} \int_0^1 (1-t)\partial_{ij}f(r+tu)dt
\end{align}
where $S_{ij} = 1$ if $i=j$ and $S_{ij} = 2$ if $i\neq j$.
\begin{align*}
    &\partial_i f(r) = \sum_{\omega \in \Omega} \frac{\omega_i}{\sqrt{p}} (-\beta_{\omega, \cos}\sin(\omega^\top r) + \beta_{\omega, \sin}\cos(\omega^\top r))\\
    &\partial_{ij}f(r+tu) = -\sum_{\omega \in \Omega} \frac{\omega_i\omega_j}{\sqrt{p}} (\beta_{\omega, \cos}\cos(\omega^\top r) + \beta_{\omega, \sin}\sin(\omega^\top r))\\
    &\forall a, b \in \RR \int_0^1 (1-t)\cos(at+b)dt =-\frac{1}{a}\sin(b) -\frac{1}{a^2}\cos(a+b)+\frac{1}{a^2}\cos(b)\\
    &\text{and} \int_0^1 (1-t)\sin(at+b)dt =\frac{1}{a}\cos(b) -\frac{1}{a^2}\sin(a+b)+\frac{1}{a^2}\sin(b)\\
    &R_{ij}(r, u) = -S_{ij}\sum_{\omega \in \Omega} \frac{\omega_i\omega_j}{\sqrt{p}}\:\bigg[ \beta_{\omega, \cos}\int_0^1 (1-t)\cos(\omega^\top(r+tu))dt \:+\:\beta_{\omega, \sin}\int_0^1 (1-t)\sin(\omega^\top(r+tu))dt \bigg]\\
    &=-S_{ij}\sum_{\omega \in \Omega} \frac{\omega_i\omega_j}{\sqrt{p}}\:\bigg[ \beta_{\omega, \cos}\big(-\frac{1}{\omega^\top u}\sin(\omega^\top r) -\frac{1}{(\omega^\top u)^2}\cos(\omega^\top u+\omega^\top r)+\frac{1}{(\omega^\top u)^2}\cos(\omega^\top r) \big)\:+\\
    &\:\beta_{\omega, \sin}\big(\frac{1}{\omega^\top u}\cos(\omega^\top r) -\frac{1}{(\omega^\top u)^2}\sin(\omega^\top u+\omega^\top r)+\frac{1}{(\omega^\top u)^2}\sin(\omega^\top r) \big)\bigg]\\
    &= -S_{ij}\frac{1}{\sqrt{p}}\sum_{\omega \in \Omega} \omega_i\omega_j[\beta_{\omega, \cos}(\cos(\omega^\top r) + c_u) + \beta_{\omega, \sin}(\sin(\omega^\top r) + s_u)]\\
    &= A(r)_{ij} + B(r)_{ij}c_u'\\
    & A(r)_{ij} = -S_{ij}\sum_{\omega \in \Omega} \frac{\omega_i\omega_j}{\sqrt{p}}\beta_{\omega, \cos}\cos(\omega^\top r) + \beta_{\omega, \sin}\sin(\omega^\top r)\\
    & B(r)_{ij} = -S_{ij}\sum_{\omega \in \Omega} \frac{\omega_i\omega_j}{\sqrt{p}}\beta_{\omega, \cos} + \beta_{\omega, \sin}
\end{align*}
with $c_u, s_u, c_u' = \mathcal{O}(|\omega^\top u|)$, i.e., there exist a constant $C$ such that $\forall u, c_u, s_u, c_u' \leq C|\omega^\top u|$. $c_u, s_u, c_u'$ are obtained by applying Taylor formulas to $\sin$ and $\cos$.

Then
\begin{align}
    |f(r+u) - f(r)| &\leq |u^\top \nabla f (r)| + |u^\top R(r+u) u|\\
    &\leq \norm{u}\norm{\nabla f (r)} + \norm{R(r+u)}_F\norm{u}^2\\
    &\leq \epsilon\norm{\nabla f (r)} + \epsilon^2\norm{R(r+u)}_F\\
    |\omega^\top u|\leq dL\epsilon
\end{align}
since $\norm{u} \leq \epsilon$.

If we have for all $r\in T$, and for all $i, j \in \llbracket 1, d\rrbracket$,
\begin{enumerate}
    \item $|f(r)| \leq \Frac{1}{3}$
    \item $|\partial_i f (r)| \leq \Frac{1}{3d\epsilon}$
    \item $|A(r)_{ij}|\leq \Frac{1}{6d\epsilon^2}$ 
    \item $|B(r)_{ij}| \leq \Frac{1}{6d^2L\epsilon^3}$
\end{enumerate}
then we will have $\forall x \in \mathcal{X}, |f(x)| \leq 1$.

Now let us lower bound the probability over drawing $\beta$s of each condition . Each time we adapt the \autoref{eqn:mcdiarmid} derived from McDiarmid's inequality 
\begin{enumerate}
    \item $\PP(|f(r)| \leq 1/3) \geq 1 - \exp(-\Frac{1}{18\sigma^2})$
    \item $\PP(|\partial_i f (r)| \leq \Frac{1}{3d\epsilon}) \geq 1 - \exp(-\Frac{1}{9d^2\epsilon^2 2\sigma^2 L^2}) = 1 - \exp(-\Frac{1}{18d^2\epsilon^2\sigma^2 L^2})$
    \item $\PP(|A(r)_{ij}| \leq \Frac{1}{6d\epsilon^2}) \geq 1 - \exp(-\Frac{1}{72d^2\epsilon^4\sigma^2 L^4})$
    \item $\PP(|B(r)_{ij}| \leq \Frac{1}{6d^2L\epsilon^3}) \geq 1 - \exp(-\Frac{1}{72d^4\epsilon^6 2\sigma^2 L^6})$
\end{enumerate}

By performing a union bound on all conditions, we have that
\begin{align}
    \PP(|f|\leq 1) &\geq \Bigg(1 - \exp(-\Frac{1}{18\sigma^2})\Bigg)^{|T|}
    \Bigg(1 - \exp(-\Frac{1}{18d^2\epsilon^2\sigma^2 L^2})\Bigg)^{d|T|}
    \Bigg(1 - \exp(-\Frac{1}{72d^2\epsilon^4\sigma^2 L^4})\Bigg)^{d^2|T|}\\
    &\Bigg(1 - \exp(-\Frac{1}{72d^4\epsilon^6 \sigma^2 L^6})\Bigg)^{d^2|T|}\\
    &\geq 1 - |T|\exp(-\Frac{1}{18\sigma^2}) - d|T|\exp(-\Frac{1}{18d^2\epsilon^2\sigma^2 L^2}) - d^2|T|\exp(-\Frac{1}{72d^2\epsilon^4\sigma^2 L^4}) - d^2|T|\exp(-\Frac{1}{72d^4\epsilon^6 \sigma^2 L^6})
\end{align}

If $\epsilon \leq \Frac{1}{2^{1/2}dL}$, then $\exp(-\Frac{1}{18d^2\epsilon^2\sigma^2 L^2}),  \exp(-\Frac{1}{72d^2\epsilon^4\sigma^2 L^4}), \text{ and }\exp(-\Frac{1}{72d^4\epsilon^6 \sigma^2 L^6}) \leq \exp(-\Frac{1}{18\sigma^2})$

Therefore
\begin{align}
    \PP(|f|\leq 1) &\geq 1 - (1+d+2d^2)\Big(Ld2^{3/2}\pi\Big)^d \exp(-\Frac{1}{18\sigma^2})
\end{align}

To satisfy these inequalities, it is enough to have $\sigma^{-1}$ on the order of $\Theta(d\:(\log d + \log L))$

\end{proof}


\end{document}